\documentclass[11pt]{article}

\input{glyphtounicode}
\usepackage[T1]{fontenc}
\usepackage[utf8]{inputenc}
\usepackage{lmodern}
\usepackage[margin=1in]{geometry}
\usepackage{setspace}
\usepackage{amsmath,amssymb,amsthm,mathtools}
\usepackage{aliascnt}
\usepackage{bm}
\usepackage{microtype}
\usepackage{enumitem}
\usepackage{booktabs}
\usepackage{graphicx}
\usepackage{float}
\usepackage{array}
\usepackage{xcolor}
\usepackage{tikz}
\usetikzlibrary{arrows.meta,positioning}
\definecolor{ectaLinkRed}{RGB}{170,0,0}
\usepackage[authoryear,round,longnamesfirst]{natbib}
\usepackage{hyperref}
\usepackage{bookmark}
\usepackage[nameinlink,noabbrev]{cleveref}

\hypersetup{
  colorlinks=true,
  linkcolor=black,
  citecolor=black,
  urlcolor=black,
  filecolor=black,
  hypertexnames=false,
  pdftitle={When Does Social Discounting Favor the Young? Welfare Comparisons in Heterogeneous Economies},
  pdfauthor={I. Sebastian Buhai},
  pdfsubject={Social discounting and welfare comparisons in heterogeneous economies},
  pdfkeywords={social discounting, welfare weights, heterogeneous agents, age comparisons, overlapping generations, common support, graph rationalizability}
}
\newtheorem{theorem}{Theorem}[section]

\newaliascnt{proposition}{theorem}
\newtheorem{proposition}[proposition]{Proposition}
\aliascntresetthe{proposition}

\newaliascnt{corollary}{theorem}
\newtheorem{corollary}[corollary]{Corollary}
\aliascntresetthe{corollary}

\newaliascnt{lemma}{theorem}
\newtheorem{lemma}[lemma]{Lemma}
\aliascntresetthe{lemma}

\theoremstyle{definition}
\newaliascnt{definition}{theorem}
\newtheorem{definition}[definition]{Definition}
\aliascntresetthe{definition}

\newaliascnt{assumption}{theorem}
\newtheorem{assumption}[assumption]{Assumption}
\aliascntresetthe{assumption}

\newaliascnt{example}{theorem}

\aliascntresetthe{example}

\theoremstyle{remark}
\newaliascnt{remark}{theorem}
\newtheorem{remark}[remark]{Remark}
\aliascntresetthe{remark}

\crefname{section}{section}{sections}
\Crefname{section}{Section}{Sections}
\crefname{subsection}{subsection}{subsections}
\Crefname{subsection}{Subsection}{Subsections}
\crefname{theorem}{theorem}{theorems}
\Crefname{theorem}{Theorem}{Theorems}
\crefname{proposition}{proposition}{propositions}
\Crefname{proposition}{Proposition}{Propositions}
\crefname{corollary}{corollary}{corollaries}
\Crefname{corollary}{Corollary}{Corollaries}
\crefname{lemma}{lemma}{lemmas}
\Crefname{lemma}{Lemma}{Lemmas}
\crefname{definition}{definition}{definitions}
\Crefname{definition}{Definition}{Definitions}
\crefname{assumption}{assumption}{assumptions}
\Crefname{assumption}{Assumption}{Assumptions}
\crefname{example}{example}{examples}
\Crefname{example}{Example}{Examples}
\crefname{remark}{remark}{remarks}
\Crefname{remark}{Remark}{Remarks}

\newcommand{\RR}{\mathbb{R}}
\newcommand{\EE}{\mathbb{E}}

\newcommand{\SWF}{\mathcal{W}}
\newcommand{\Age}{\mathcal{A}}
\newcommand{\State}{\mathcal{Z}}
\newcommand{\Pop}{\mathcal{X}}

\newcommand{\CompSpace}{\mathcal{U}}
\newcommand{\Mterm}{\mathcal{M}}
\newcommand{\Rterm}{\mathcal{R}}
\newcommand{\Qterm}{\mathcal{Q}}
\newcommand{\PVEU}{\mathrm{PVEU}}

\newcommand{\dist}{\mu}
\newcommand{\ndw}{\eta}
\newcommand{\smw}{\omega}
\newcommand{\contv}{\upsilon}
\newcommand{\Comp}{\chi}
\newcommand{\Spread}{\mathfrak{G}}
\newcommand{\DemDist}{\mathfrak{D}}

\newcommand{\WedgeTerm}{\Lambda}
\newcommand{\BridgeCorr}{\mathfrak{B}}
\newcommand{\CWedgeTerm}{\widetilde{\Lambda}}
\newcommand{\dd}{\mathrm{d}}

\newcommand{\esssup}{\operatorname*{ess\,sup}}
\newcommand{\essinf}{\operatorname*{ess\,inf}}

\newcommand{\BaselineMatchMean}{4.40\times 10^{-4}}

\newcommand{\BaselineOverlapOverall}{1.000000}

\newcommand{\BaselineStrictInteriorShare}{0.8440}
\newcommand{\BaselineBoundaryLoadingShare}{0.1560}
\newcommand{\BaselineImplementedWedgeMean}{1.0119}
\newcommand{\BaselineImplementedWedgePNinetyFive}{1.0839}

\newcommand{\LiquidIlliquidExactOverlap}{0.8151}
\newcommand{\LiquidIlliquidTercileOverlap}{0.8946}

\newcommand{\LiquidIlliquidTercileMeanKappaPlp}{6.71}
\newcommand{\LiquidIlliquidTercileMaxKappaPlp}{28.47}
\newcommand{\LiquidIlliquidTercilePointDelta}{-0.2706}
\newcommand{\LiquidIlliquidEarningsOnlyPointDelta}{1.7445}
\newcommand{\LiquidIlliquidAgeOnlyPointDelta}{0.9560}
\newcommand{\LiquidIlliquidTercileGraphLower}{-3.2369}
\newcommand{\LiquidIlliquidTercileGraphUpper}{6.7702}
\newcommand{\LiquidIlliquidEarningsOnlyGraphLower}{-4.1160}
\newcommand{\LiquidIlliquidEarningsOnlyGraphUpper}{7.9730}
\newcommand{\LiquidIlliquidAgeOnlyGraphLower}{-1.5240}
\newcommand{\LiquidIlliquidAgeOnlyGraphUpper}{6.9176}

\newcommand{\MarginAuditMaxDeltaLogQ}{0.00000}
\newcommand{\MarginAuditRowMismatchShare}{0.0000}
\newcommand{\MarginAuditMeanDeltaLogLambda}{0.00092}
\newcommand{\MarginAuditPninetyFiveDeltaLogLambda}{0.00649}

\newcommand{\draftnote}[1]{}

\newcommand{\PublicTitleNote}{%
All replication files are available in the project's GitHub repository at \url{https://github.com/sbuhai/current-cell-matching-project}. The repository contains the code, generated table sources, and input files needed to reproduce the quantitative application. No restricted data are used. I have used \href{https://refine.ink}{Refine.ink}, developed by Ben Golub and Yann Calv\'o L\'opez, to assist with notational consistency and proof clarity. An earlier version of this paper was entitled "Who Counts as Young? Matching Rules for Social Discounting".%
}

\title{When Does Social Discounting Favor the Young? \\ Welfare Comparisons in Heterogeneous Economies\thanks{\PublicTitleNote}}
\author{%
I. Sebastian Buhai\thanks{Email: \href{mailto:sebastian.buhai@sofi.su.se}{sebastian.buhai@sofi.su.se}. Website: \href{https://www.sebastianbuhai.com}{sebastianbuhai.com}.}\\[0.4em]
SOFI, Stockholm University\\
Instituto de Econom\'ia, UC Chile\\
NIPE, University of Minho}
\date{First version: June 2, 2026. This version: July 11, 2026. \href{https://www.sebastianbuhai.com/papers/publications/current_cell_matching.pdf}{Latest version}.}

\makeatletter
\let\WP@maketitle\maketitle
\renewcommand{\maketitle}{%
  \begingroup
    \renewcommand{\@fnsymbol}[1]{\@arabic{##1}}%
    \WP@maketitle
  \endgroup
  \setcounter{footnote}{2}
}
\makeatother

\begin{document}

\maketitle

\begin{abstract}
Social discounting is often interpreted as ranking young and old agents. In a heterogeneous economy, however, a social discount schedule does not determine a unique age ranking. The sign can reverse with the states being compared and may remain unidentified when common support is disconnected. I characterize when local young-old comparisons isolate age-related social priority rather than differences in private marginal values, constraints, or units. A coherent system of normalized welfare weights exists if and only if every supported path between the same states implies the same ratio. In a calibrated life-cycle economy, every retained local comparison favors the young, yet the aggregate ranking favors the old under one component normalization, favors the young on the largest shared component, and is unidentified when component scales are unrestricted.

\medskip
\noindent\textbf{Keywords.} Social discounting; welfare weights; heterogeneous agents; age comparisons; overlapping generations; common support; graph rationalizability

\smallskip
\noindent\textbf{JEL codes.} D61, D63, D15, D52, E21, H43, C63
\end{abstract}

\clearpage

\section{Introduction}\label{sec:introduction}

Does a lower social discount rate imply that society places greater weight on the young? In a representative-agent model, or in an overlapping-generations model in which age is the only state, the question has a well-defined answer. In a heterogeneous economy it does not. A marginal social welfare weight is attached to an economic state, not to an age label. Young and old agents differ in assets, earnings, constraints, health, and private marginal values. The inferred age ranking therefore depends on which states are compared and on the unit in which the comparison is made.

This paper establishes three results. First, a social discount schedule alone does not identify a welfare ranking of young and old agents. Coarse state matching can reverse the sign of the normalized young-old comparison even when the discount schedule and the within-person intertemporal tradeoff are held fixed. Second, supported local comparisons admit one coherent system of normalized welfare weights only when every route between the same states implies the same ratio. Third, disconnected support identifies those weights only within connected components. An aggregate age ranking then requires a normalization across components, and its sign can change with that normalization.

The key decomposition separates social valuation from private marginal value. Let $\smw_t(x)$ denote the planner's marginal value of one current resource unit at state $x$, and let $\contv_t(x)$ denote the individual's private marginal value of that unit. Their ratio
\[
\ndw_t(x)=\frac{\smw_t(x)}{\contv_t(x)}
\]
is a normalized welfare weight. Comparing $\ndw_t$ across ages isolates social priority only when the states also agree on the non-age information needed to price the policy margin. For one-period saving, that information determines the continuation payoff, market compensation, and active constraint wedge. The matching rule is therefore derived from the policy margin rather than imposed by the age labels.

Relative to \citet{Eden2023}, the contribution is an exact characterization of when the cross-sectional age implication of social discounting survives heterogeneity. Eden's benchmark is the right starting point once the relevant young-old comparison is defined. In a heterogeneous economy it becomes a family of possible comparisons. Age matching isolates age-related social priority only for states that share the policy-relevant current information and the same private marginal value. Otherwise the comparison remains economically meaningful, but it combines age with state or unit differences.

Each supported young-old comparison implies a ratio of normalized welfare weights. The ratio combines the social discount schedule, the market compensation factor for the policy margin, and any constraint or continuation-value correction. These local ratios can be aggregated only if they are path independent. All supported routes between the same two states must agree, or equivalently the implied ratios must multiply to one around every closed chain. When this condition holds, normalized welfare weights are identified up to one positive constant in each connected component. The network representation makes the identification failure from disconnected support explicit.

The calibrated application shows that these distinctions change the economic conclusion. The focal comparison contrasts ages $25$ to $34$ with ages $45$ to $54$ at a common current-earnings state. Age-only comparisons do not identify a sign. Matching states on the one-period saving margin and fixing the normalization across components yields an interval from $-0.00819$ to $-0.00298$, so the older target receives the larger normalized weight. On the largest component shared by both targets, the interval is instead positive. If component normalizations are unrestricted, the aggregate comparison is unidentified. In the application, every retained local comparison favors younger agents, yet the aggregate ranking is negative under one cross-component normalization and positive within the shared component.

The sign change is an identification result rather than a numerical artifact. The finite program is feasible, and its duplicate-link, cycle, band, and relaxation residuals are zero at the reported precision. Support and normalization, not numerical failure, account for the different rankings. Alternative private units, constraint treatments, and richer state spaces show which conclusions are tied to the maintained welfare object.

The issue extends beyond age. Welfare comparisons between demographic groups compare heterogeneous states. Whenever a group label suppresses state variables that price the policy margin, raw welfare weights combine group priority with private marginal value, constraints, and support. The results therefore provide a general test for whether local group comparisons share a coherent welfare interpretation.

The paper is closest to two literatures. Work on social discounting studies how intergenerational weights shape cross-sectional and dynamic welfare comparisons \citep{CaplinLeahy2004,FarhiWerning2007,FengKe2018,Eden2023,MillnerHeal2023}. I characterize the conditions under which the age comparison used in that literature remains identified with heterogeneous states. Welfare-accounting methods evaluate specified perturbations under a chosen money metric or welfare-weight system \citep{BaqaeeBursteinKoikeMoriQJE2024,DavilaSchaab2025,BarconsDavilaSchaab2026}. Here the prior question is which states, unit, and support make an age comparison a valid input into that exercise. The common private scale is an interpersonal-comparability convention in the sense of \citet{Harsanyi1955}, while the support requirement parallels common-support restrictions in matching and program evaluation \citep{HeckmanIchimuraTodd1997,ImbensWooldridge2009}.

\Cref{sec:benchmark} states the main economic results before the general environment. \Cref{sec:environment,sec:welfare} define local welfare weights and comparable states. \Cref{sec:stochastic-mr,sec:nonpaternalism,sec:age-blindness} derive the local decomposition, market valuation, and coherence restrictions. \Cref{sec:illustrations} presents the calibrated application. Technical extensions and proofs are collected in the appendices.


\section{Why a Social Discount Rate Does Not Rank Ages}\label{sec:benchmark}

The identification problem is visible before the general model is introduced. A social discount schedule compares a current old-age resource unit with the future old-age payoff reached by a currently young person. To turn that intertemporal comparison into a current young-old ranking, one must also specify the current states being compared and the unit used across them.

\subsection{Three welfare objects}

Let $\smw_t(x)$ be the social value of one current resource unit, $\contv_t(x)$ the private marginal value of that unit, and $\ndw_t(x)=\smw_t(x)/\contv_t(x)$ the normalized welfare weight. These objects answer different questions. \Cref{tab:object-taxonomy} separates them.

\begin{table}[htbp]
\centering
\caption{Three local welfare comparisons}
\label{tab:object-taxonomy}
\footnotesize
\renewcommand{\arraystretch}{1.08}
\resizebox{0.98\textwidth}{!}{%
\begin{tabular}{>{\raggedright\arraybackslash}p{0.23\textwidth}>{\raggedright\arraybackslash}p{0.30\textwidth}>{\raggedright\arraybackslash}p{0.20\textwidth}>{\raggedright\arraybackslash}p{0.27\textwidth}}
\toprule
Object & Comparison & What is held fixed? & Interpretation \\
\midrule
Resource welfare weight $\smw$ & Social value of one unit of current goods & Nothing beyond the state itself & Combines social priority with private marginal value. \\
Normalized welfare weight $\ndw=\smw/\contv$ & Social value per unit of private marginal value & A common private scale & Removes private marginal-value levels but remains indexed by the chosen scale. \\
Matched resource comparison & Resource weights across ages & Policy-relevant current information and $\contv$ & Isolates age-related social priority for the maintained policy margin and support. \\
\bottomrule
\end{tabular}%
}
\end{table}

The third object is the focus of the paper. A comparison is exact when the young and old states share the non-age information needed to value the policy margin and have the same private marginal value. I call the corresponding set a matched set, and use the formal term comparison fiber after the general environment is introduced. The economic requirement is that the comparison hold fixed everything that would otherwise change the value of the marginal experiment apart from age.

\subsection{State matching and sign reversal}

A simple example shows why age alone is insufficient. Let $o$ be an old state and let $y_L,y_H$ be two young states of the same age. Suppose
\[
(\smw_t,\contv_t,\ndw_t)(o)=(6,2,3),\qquad
(\smw_t,\contv_t,\ndw_t)(y_L)=(8,4,2),\qquad
(\smw_t,\contv_t,\ndw_t)(y_H)=(4,1,4).
\]
In resource units, matching $o$ with $y_L$ favors the young, while matching $o$ with $y_H$ favors the old. In normalized units, the first comparison favors the old and the second favors the young. The ranking changes before any ethical judgment about age is altered.

For a local policy margin, let $\mathcal D^S$ denote the social discount ratio and let $\Rterm^S$ denote the within-person intertemporal comparison along the younger agent's continuation. When the current states are exactly matched,
\begin{equation}\label{eq:preview-exact-decomposition}
\log\frac{\ndw_t(x_t^{\mathrm y})}{\ndw_t(x_t^{\mathrm o})}
=
\log\Rterm^S-\log\mathcal D^S.
\end{equation}
A coarser comparison adds the private-value gap,
\begin{equation}\label{eq:preview-coarse-decomposition}
\log\frac{\ndw_t(x_t^{\mathrm y})}{\ndw_t(x_t^{\mathrm o})}
=
\log\Rterm^S-\log\mathcal D^S
+
\log\frac{\contv_t(x_t^{\mathrm o})}{\contv_t(x_t^{\mathrm y})}.
\end{equation}
The additional term is not approximation error. It is a different welfare object.

Under the baseline within-person valuation restriction derived in \Cref{sec:nonpaternalism}, $\Rterm^S=\Qterm\WedgeTerm$, where $\Qterm$ is market compensation for the policy margin and $\WedgeTerm\ge1$ is the constraint wedge. Exact matching therefore gives
\[
\log\frac{\ndw_t(x_t^{\mathrm y})}{\ndw_t(x_t^{\mathrm o})}
=
\log\Qterm-\log\mathcal D^S+\log\WedgeTerm.
\]

\begin{corollary}[When lower social discounting favors the young]\label{cor:trilemma}
Fix a supported system of exactly matched comparisons, with the induced measure defined in Appendix~\ref{app:technical-foundations}. The following statements cannot all hold: the social discount ratio is no greater than market compensation on every comparison and is strictly smaller on a positive-measure set; the planner respects the individual's ranking of the within-person payoff; and normalized welfare weights are equal across the matched young and old states. If the latter two statements hold, then $\mathcal D^S=\Qterm\WedgeTerm\ge\Qterm$ almost everywhere.
\end{corollary}

Thus discounting below market compensation favors the younger state under exact matching. The next result shows why this conclusion does not survive a coarser age comparison.

\begin{proposition}[Coarse matching can reverse an age ranking]\label{prop:coarse-ambiguity}
For any $\gamma>0$ and $\delta>\gamma$, there is an admissible local environment and a coarse current comparison with common policy-relevant state information such that
\[
\log \Rterm^S-\log\mathcal D^S=\gamma,
\qquad
\log\frac{\contv_t(x_t^{\mathrm o})}{\contv_t(x_t^{\mathrm y})}=-\delta,
\]
and therefore
\begin{equation}\label{eq:coarse-ambiguity}
\log\frac{\ndw_t(x_t^{\mathrm y})}{\ndw_t(x_t^{\mathrm o})}=\gamma-\delta<0.
\end{equation}
An exact comparison with the same $\Rterm^S$ and $\mathcal D^S$ instead gives the positive log ratio $\gamma$.
\end{proposition}

The proposition isolates the source of reversal. A social discount schedule and a within-person intertemporal comparison can imply greater normalized weight on the young under exact matching, while the same objects imply the opposite ranking after private marginal values are allowed to differ.

\subsection{Coherent welfare weights}

Represent each supported young-old comparison by a directed link $e=(x^{\mathrm y},x^{\mathrm o})$. Different continuation states or market implementations can support the same current comparison. Let $\Pi_e$ collect those verified implementations. Each $\pi\in\Pi_e$ implies the log ratio
\[
\bar g(\pi)=\log\Qterm(\pi)-\log\bar{\mathcal D}^S(\pi)+\log\CWedgeTerm(\pi),
\]
where $\Qterm$ is market compensation for the policy margin and $\CWedgeTerm$ collects the active constraint wedge and any declared continuation adjustment. The next theorem gives the condition under which these local ratios belong to one welfare-weight system.

\begin{theorem}[Coherent normalized welfare weights]\label{prop:imposed-discount-rationalization}
Fix a policy margin, the state information held fixed across ages, a common private unit, and a within-person valuation restriction. Construct the finite supported comparison network. For each directed link $e=(x^{\mathrm y},x^{\mathrm o})$, let $\Pi_e$ be the set of verified implementations. An imposed local schedule $\bar{\mathcal D}^S$ is consistent with positive normalized welfare weights if and only if two conditions hold. First, all implementations of the same link imply the same value of $\bar g$. Second, every closed chain has zero signed label sum, counting a link traversed in reverse with the opposite sign. When these conditions hold, there exists $\widetilde\eta>0$ such that
\[
\log\widetilde\eta(x^{\mathrm y})-\log\widetilde\eta(x^{\mathrm o})=\bar g(e)
\]
for every edge. The solution is unique up to one multiplicative constant in each connected component. Relative levels across disconnected components, and target comparisons that load differently on those components, are not identified without additional normalizations.
\end{theorem}

The theorem separates coherence from normalization. Closed-chain consistency determines whether the local comparisons come from one welfare-weight system. Connected support identifies every ratio within a component. Local comparisons alone cannot identify the relative scale of disconnected components.

\subsection{Quantitative implication}

The calibrated application makes the distinction visible. Target $A$ contains ages $25$ to $34$ in the central current-earnings state, and target $B$ contains ages $45$ to $54$ on the same one-period saving margin. Positive values favor the younger target. \Cref{tab:main-diagnostic-audit} reports the central result.

\begin{table}[htbp]
\centering
\caption{How matching and support change the age ranking}
\label{tab:main-diagnostic-audit}
\footnotesize
\renewcommand{\arraystretch}{1.06}
\begin{tabular}{>{\raggedright\arraybackslash}p{0.43\textwidth}cc}
\toprule
Maintained comparison & Identified interval & Implication \\
\midrule
Age only, full support & $[-0.04044,\ 0.15093]$ & No age ranking \\
Matched states, fixed component normalization & $[-0.00819,\ -0.00298]$ & Favors the older target \\
Matched states, unrestricted component normalization & $[-\infty,\ +\infty]$ & Unidentified \\
Largest component shared by both targets & $[0.00004,\ 0.00082]$ & Favors the younger target \\
\bottomrule
\end{tabular}
\begin{minipage}{0.94\textwidth}
\footnotesize Notes. Intervals are normalized welfare-weight contrasts. Positive values favor target $A$, the younger group. The fixed-normalization and largest-component rows apply the same local comparison rule but differ in the support over which levels are compared.
\end{minipage}
\end{table}

The age-only comparison is uninformative. Exact matching recovers informative local ratios, but aggregation still depends on the support that connects the two targets. The fixed-normalization result is negative because the targets load differently across components, even though every retained local comparison tilts toward younger agents. The shared component gives the opposite sign, and unrestricted component scales leave the aggregate ranking unidentified. A social discount rate therefore does not rank ages until the state match, unit, support, and cross-component normalization have been specified.

For the one-period saving application, the matched state information is the current earnings transition row, together with the private marginal value of current resources. \Cref{sec:nonpaternalism} derives this rule from the household problem. Appendix~\ref{app:technical-foundations} records analogous verification requirements for risky assets, labor margins, and liquid-illiquid saving.


\section{Environment}\label{sec:environment}

The formal analysis needs a cross section of heterogeneous agents, local resource perturbations, and the continuation opportunities available to each current state. This section defines those objects.

\subsection{Demographic structure and uncertainty}

Time is discrete, $t\in\mathbb Z_+$. The benchmark aggregate path is deterministic; uncertainty indexes idiosyncratic histories. Each period a cohort is born. Ages are $a\in\Age=\{0,1,\ldots,A\}$. An individual's state is $z\in\State$, a Borel subset of Euclidean space. This may include assets, productivity, health, promised utility, or any state needed for private choices and social comparisons. A state point is $x=(a,z)\in\Pop:=\Age\times\State$. I use \emph{cell} as shorthand when referring to a point on a finite grid.

At date $t$, the living cross section is the probability measure $\dist_t$ on $\Pop$. The benchmark path is $\{\dist_t,c_t,p_t\}_{t\ge0}$, where $c_t(x)$ is the private allocation and $p_t$ collects prices, taxes, transfers, and aggregate objects.

\subsection{Private behavior and continuation opportunities}

Individuals optimize given benchmark prices, policies, and constraints. For $s\ge t$ and current state $x$, let $\Gamma_{t,s}(x)\subseteq\Pop$ be the support of the date-$s$ state conditional on survival to $s$, and set it to $\emptyset$ when survival has zero probability. Survival probabilities are included in the weights on future payoffs.

For $y\in\Gamma_{t,s}(x)$, the attainable set
\[
\mathcal T_{t,s}(x,y)\subseteq\RR_+^2
\]
contains pairs $(\delta_t,\delta_s)$ for which surrendering $\delta_t$ units at date $t$ can finance $\delta_s$ units at date $s$ along the same individual history. The comparison is supported when the maximal exchange ratio is finite and positive and the attainable correspondence is measurable.

\begin{assumption}[Admissible benchmark path]\label{ass:benchmark}
The benchmark path is feasible, measurable, and implementable. Conditional on age, $z$ is sufficient for the local continuation problem; $\Gamma_{t,s}$ is measurable for every $s\ge t$; and the benchmark allocation solves the individual's problem given prices and constraints.
\end{assumption}

\subsection{Local perturbations}

A benchmark perturbation path is a feasible family of resource additions $m_t^\tau:\Pop\to\RR$ around the benchmark, with fixed prices, policies, survival laws, transition laws, and $\dist_t$ unless an extension declares a general-equilibrium object. Its first-order signed resource shift is $\Delta=\{\Delta_t\}_{t\ge0}$.

\begin{definition}[Admissible perturbation]\label{def:perturbation}
A signed finite-horizon perturbation $\Delta$ is \emph{admissible} if:
\begin{enumerate}[label=(\roman*), leftmargin=2.3em]
    \item $\Delta_t$ is absolutely continuous with respect to $\dist_t$, with density $h_t\in L^\infty(\dist_t)$ for every $t$;
    \item $\Delta$ is the two-sided first-order resource shift generated by a feasible benchmark perturbation path; and
    \item the relevant social and private G\^ateaux derivatives along $\Delta$ exist.
\end{enumerate}
\end{definition}

On finite grids, local weights are derivatives per unit of cell mass. Appendix~\ref{app:technical-foundations} states the corresponding localization and version conventions for continuous supports.

\subsection{Valuation systems}

Social valuation is the directional derivative of the planner's objective. Private valuation is summarized by the marginal value $\contv_t(x)$ of current resources and, for within-person transfers, by market compensation or supporting state prices. Once a common private unit is fixed, $\ndw_t(x)=\smw_t(x)/\contv_t(x)$ is the normalized welfare weight. The remaining question is when ratios of this object compare like with like.


\section{Comparable Welfare Weights}\label{sec:welfare}

Comparing welfare weights across ages requires a common private scale and a rule for holding policy-relevant state information fixed. This section defines both. The resulting matched sets isolate age-related social priority in resource units. Coarser sets also load differences in private marginal value or in the state variables that price the policy margin.

\subsection{Local derivatives}

\begin{assumption}[Local differentiability and positivity]\label{ass:local-derivatives}
For every admissible perturbation $\Delta$, write $\Delta_t(\dd x)=h_t(x)\dist_t(\dd x)$. The social objective admits the G\^ateaux derivative
\[
D\SWF[\Delta]=\sum_{t\ge0}\int_{\Pop}\smw_t(x)h_t(x)\dist_t(\dd x)
\]
for some measurable $\smw_t:\Pop\to(0,\infty)$, defined $\dist_t$ almost everywhere, with finite displayed integrals on the admissible bounded-density domain. The tangent space is dated-separating: if $\sum_t\int f_t h_t\,\dd\dist_t=0$ for every admissible finite-horizon perturbation, then $f_t=0$ $\dist_t$ almost everywhere on each active date. For each active $t$ and for $\dist_t$-almost every cell $x$, the private continuation value $V_t(x;m)$, evaluated after adding $m$ units of the current resource at $x$, is differentiable at $m=0$, with
\[
\contv_t(x):=\left.\frac{\partial V_t(x;m)}{\partial m}\right|_{m=0}\in(0,\infty).
\]
\end{assumption}

\begin{definition}[Social marginal welfare weight]\label{def:smw}
The \emph{social marginal welfare weight} is any measurable representer $\smw_t(x)$ satisfying, for every admissible $\Delta$,
\begin{equation}\label{eq:smw}
D\SWF[\Delta]=\sum_{t\ge0}\int_{\Pop}\smw_t(x)h_t(x)\dist_t(\dd x).
\end{equation}
It is unique up to $\dist_t$-null sets by dated separation.
\end{definition}

\begin{definition}[Normalized welfare weight]\label{def:ndw}
Under \Cref{ass:benchmark,ass:local-derivatives}, define
\begin{equation}\label{eq:ndw}
\ndw_t(x):=\frac{\smw_t(x)}{\contv_t(x)}.
\end{equation}
\end{definition}

\begin{lemma}[Factorization and scale properties]\label{lem:factorization}
On $\{\contv_t>0\}$,
\begin{equation}\label{eq:factorization}
\smw_t(x)=\ndw_t(x)\contv_t(x).
\end{equation}
If the private utility index is transformed everywhere by $u\mapsto\alpha u+\beta$, $\alpha>0$, while the planner's cardinal representation over physical allocations is fixed, then $\contv_t$ scales by $\alpha$, $\smw_t$ is unchanged, and $\ndw_t$ scales by $\alpha^{-1}$. Ratios of $\ndw$ across states are invariant to a common positive affine rescaling.
\end{lemma}

\begin{proof}
The identity is the definition of $\ndw_t$. A common rescaling multiplies every private marginal value by the same $\alpha$ and leaves the social derivative with respect to physical resources unchanged, so ratios of normalized welfare weights are unchanged.
\end{proof}

\subsection{The maintained private unit}

\begin{assumption}[Common cardinal private scale]\label{ass:common-cardinal}
Comparisons of $\contv_t(x)$ across states use a common cardinal index for the private continuation problem. Common positive affine transformations are innocuous. Slope changes that vary across states are not.
\end{assumption}

I maintain the private unit as an interpersonal-comparability convention in the sense of \citet{Harsanyi1955}. The physical perturbation is one unit of the current good, so $\smw_t(x)$ is its marginal social value in the planner's feasibility constraint. Every result is indexed by the policy margin $M$, the state-matching map $\Comp$, and the private unit $\contv$. A common positive affine rescaling of the private index cancels from ratios of $\ndw$. A slope that varies across states changes which private marginal values are equal and therefore changes the matched sets, support network, and normalization across disconnected components. It defines a different welfare comparison, not a harmless change of numeraire.

\subsection{Welfare weights in private-value units}

One private-value-equivalent unit at state $x$ is a resource shift of size $1/\contv_t(x)$. It raises private continuation value by one marginal unit, and $\ndw_t(x)$ is its social value. This change of units permits comparisons when private marginal values differ. It does not by itself give a resource-unit comparison of age priority. Appendix~\ref{app:technical-foundations} gives the formal measure representation.

\subsection{Matching states for a policy margin}

For policy margin $M$, let $Z_t^M(x)$ collect the non-age current information that identifies the induced future payoff, market compensation, and active constraints. Deterministic schedules that depend on age remain part of the age comparison.

\begin{definition}[State-matching map for a policy margin]\label{def:margin-sufficient}
Fix a local policy margin $M$ at date $t$. A measurable map
$\Comp_t:\Pop\to\CompSpace$ is \emph{margin-sufficient} for $M$ if it applies the
same coding of non-age information across age slices and, whenever $x=(a,z)$ and
$x'=(a',z')$ satisfy $\Comp_t(x)=\Comp_t(x')$ and $\contv_t(x)=\contv_t(x')$ in a
verified comparison, those equalities, together with the age alignment, identify
all non-age objects needed to determine the admissibility and valuation of the
future payoff induced by the policy margin, the market factor $\Qterm$, and the current loadings entering
$\WedgeTerm$.
\end{definition}

A verification certificate records the economic content of this condition. It lists the future payoff, price object, and active constraints, then retains the non-age state information needed to identify them once age alignment and $\contv_t$ are fixed. Appendix~\ref{app:technical-foundations} gives certificates for saving, risky assets, labor taxes, and liquid-illiquid portfolios. Any coarsening of $\Comp_t$ requires a new sufficiency check.

\begin{proposition}[Verified margin sufficiency]\label{prop:margin-sufficient-map}
Fix $M$ and suppose that, on the maintained support, $\Comp_t^M$ is a verified
margin-sufficient map that is minimal in the partition-refinement order among maps
preserving the admissible pair relation at date $t$.
\begin{enumerate}[label=(\roman*), leftmargin=2.3em]
    \item Matching within the exact sets defined by $(\Comp_t^M,\contv_t)$ holds fixed all current state
    objects, apart from age, that are needed to interpret $M$. Deterministic schedules
    that depend on age may vary.
    \item If a coarsening $\widetilde\Comp_t$ merges cells of the same age with equal
    $\contv_t$ but different verified codes, sufficiency must be checked again.
    \item Every refinement $\bar\Comp_t=(\Comp_t^M,W_t)$ remains margin-sufficient and
    can only reduce the set of admissible matches.
\end{enumerate}
\end{proposition}

\begin{definition}[Matched set (comparison fiber)]\label{def:fiber}
For a verified margin-sufficient map $\Comp_t$, date $t$, attribute $u\in\CompSpace$, and private envelope scale $v>0$, the \emph{comparison fiber} is
\begin{equation}\label{eq:fiber}
\mathcal F_t(u,v):=\{x\in\Pop:\Comp_t(x)=u \text{ and } \contv_t(x)=v\}.
\end{equation}
\end{definition}

States in a verified matched set share the policy-relevant current information and the same private marginal value. Ratios of $\ndw$ within that set therefore compare age-related social priority for transfers measured in physical resources. A coarser map has this interpretation only after a separate sufficiency check.

\begin{proposition}[Changing the private unit changes the comparison]\label{prop:common-scale-scope}
Under \Cref{ass:benchmark,ass:local-derivatives,ass:common-cardinal}, a common positive affine rescaling leaves every ratio of normalized welfare weights unchanged. A rescaling that varies across states can change $\contv_t$, the matched sets, supported comparisons, connected components, component normalizations, and target intervals. Alternative private-unit conventions therefore require recomputing the comparison network.
\end{proposition}

\begin{proof}
A common rescaling multiplies all private marginal values by the same constant and cancels from ratios of normalized welfare weights. A rescaling that varies across states changes the equality relations in \Cref{eq:fiber} and therefore every object built from the matched sets.
\end{proof}

\begin{lemma}[Minimal matched sets for age priority]\label{lem:exact-fiber-characterization}
Fix $M$ and a verified margin-sufficient map $\Comp_t^M$. Among rules that use only
$(\Comp_t^M,\contv_t)$, the matched sets are the coarsest partition that keeps fixed both
the verified code and the private unit. On this partition, ratios of $\smw_t$ coincide
with ratios of $\ndw_t$, the rule is invariant to a common rescaling of $\contv_t$, and
any coarser rule must either merge distinct verified codes or compare different private
units.
\end{lemma}

\begin{proof}
If a rule admits a pair with unequal $\contv_t$, the resource-weight ratio differs from the normalized ratio by $\contv_t(x)/\contv_t(x')$. If it merges distinct verified codes, it abandons the verified sufficiency of the state match for $M$. The matched set imposes exactly these two equalities and is unchanged by a common rescaling.
\end{proof}

\begin{remark}[Weaker comparability]\label{rem:weaker-comparability}
With only partial cardinal comparability, the comparison network is indexed by the weaker unit system. All network and interval statements remain conditional on that unit.
\end{remark}

\begin{lemma}[Nonvacuity under a monotone matching state]\label{lem:fiber-nonempty}
Suppose that for two ages and a common verified code $u$, the scalar resource state $b\mapsto\contv_t(a,b,u)$ is continuous and strictly monotone on support intervals whose images overlap. Then every $v$ in the overlap defines a nonempty cross-age exact fiber.
\end{lemma}

\subsection{Local age neutrality}

\begin{definition}[Conditional age neutrality]\label{def:cab}
Given a comparison system $(M,\Comp,\contv)$, the social criterion is \emph{conditionally residual-age neutral} on a set of exactly matched comparisons if
\[
\ndw_t(x_t^{\mathrm o})=\ndw_t(x_t^{\mathrm y})
\]
for each comparison in the set. The neutrality statement is local to the policy margin, state-matching map, private unit, and support.
\end{definition}

\begin{remark}[From exact comparisons to reported intervals]\label{rem:hierarchy}
Exactly matched comparisons are the objects governed by the theorem. An $\varepsilon$ band turns
equality into explicit mismatch bounds, and a coarser map adds a contamination radius
$\kappa$. Finite grid recursions then choose point summaries inside these sets. The
numerical section reports intervals before any point summary.
\end{remark}


\section{Local Age Comparisons}\label{sec:stochastic-mr}

The local comparison separates age-related social priority from the younger agent's within-person intertemporal tradeoff. Exact matching removes the private-value term. Coarser matching retains it, and approximate matching bounds it explicitly. \Cref{sec:age-blindness} then asks whether the resulting local ratios can be combined into one welfare-weight system.

\subsection{The local comparison}

Fix date $t$. Let $x_t^{\mathrm o}$ and $x_t^{\mathrm y}$ be current states in the same matched set, with $a^{\mathrm o}>a^{\mathrm y}$, and set $s=t+a^{\mathrm o}-a^{\mathrm y}$. A comparison also records the state reached by the younger person at date $s$, when that person has the older agent's current age.

\begin{definition}[Local age comparison]\label{def:comparison-triple}
A \emph{local age comparison} is
\[
\pi=(t,x_t^{\mathrm o},x_t^{\mathrm y},x_s^{\mathrm o})
\]
such that (i) $x_t^{\mathrm o}$ and $x_t^{\mathrm y}$ lie at date $t$ and in the same comparison fiber; (ii) $a(x_t^{\mathrm o})>a(x_t^{\mathrm y})$; and (iii) $x_s^{\mathrm o}\in\Gamma_{t,s}(x_t^{\mathrm y})$ with $a(x_s^{\mathrm o})=a(x_t^{\mathrm o})$.
\end{definition}

\begin{definition}[Ordered, approximate, and verified comparisons]\label{def:ordered-bridge-tuple}
An \emph{ordered comparison} is a quadruple $(t,x_t^{\mathrm o},x_t^{\mathrm y},x_s^{\mathrm o})$ satisfying the same age ordering, alignment, and continuation condition without imposing an exact match. It is \emph{approximately matched} when the maintained state codes agree and the stated bound on the private marginal-value difference holds. It is \emph{verified} when the induced future payoff, supporting-price/KKT representation, and within-person valuation restriction apply. A nonbaseline comparison also carries the positive continuation adjustment $\BridgeCorr(\pi)$ introduced in \Cref{subsec:bridge-assumptions}.
\end{definition}

\begin{remark}[Future payoffs with several states]\label{rem:portfolio-legs}
The displayed comparison uses one continuation state. For a future payoff with finite support $L_s^\pi=\sum_j\ell_j\delta_{x_{sj}}$, with survival and transition probabilities absorbed in $\ell_j$, write
\[
\mathcal V_s^\omega(\pi):=\sum_j\ell_j\smw_s(x_{sj}),
\]
and set $\mathcal V_s^\omega(\pi)=\smw_s(x_s^{\mathrm o})$ in the singleton case. The corresponding $\Qterm(\pi)$ is the gross market compensation factor for that payoff, and the one-period risk-free margin uses the same convention.
\end{remark}

\begin{definition}[Gross social discount ratio]\label{def:sdf-social}
For $\pi=(t,x_t^{\mathrm o},x_t^{\mathrm y},x_s^{\mathrm o})$, or the corresponding portfolio version, define
\begin{equation}\label{eq:Ds}
\mathcal D^S(\pi):=\frac{\smw_t(x_t^{\mathrm o})}{\mathcal V_s^\omega(\pi)}.
\end{equation}
It is the number of units of the future old-age payoff that is socially equivalent, at the margin, to one current old-age resource unit.
\end{definition}

Under this convention, $\mathcal D^S(\pi)$ is the gross social discount ratio for the future payoff, so $\rho^S$ is the corresponding net rate when $\mathcal D^S=1+\rho^S$. The factor $\Qterm(\pi)$ is market compensation for the same payoff, and $\WedgeTerm(\pi)$ records the effect of active constraints. The social and market ratios must value the same payoff in the same private unit.

\subsection{Decomposition under exact matching}

\begin{proposition}[Welfare-weight decomposition under exact matching]\label{thm:mr}
For any local age comparison $\pi$, define
\begin{equation}\label{eq:Mterm}
\Mterm(\pi):=\frac{\ndw_t(x_t^{\mathrm o})}{\ndw_t(x_t^{\mathrm y})}
\end{equation}
and
\begin{equation}\label{eq:Rsocial}
\Rterm^S(\pi):=\frac{\smw_t(x_t^{\mathrm y})}{\mathcal V_s^\omega(\pi)}.
\end{equation}
Then
\begin{equation}\label{eq:mr-decomposition}
\mathcal D^S(\pi)=\Mterm(\pi)\Rterm^S(\pi).
\end{equation}
\end{proposition}

\begin{proof}
Using $\smw_t=\ndw_t\contv_t$,
\[
\mathcal D^S(\pi)=
\frac{\ndw_t(x_t^{\mathrm o})}{\ndw_t(x_t^{\mathrm y})}
\frac{\contv_t(x_t^{\mathrm o})}{\contv_t(x_t^{\mathrm y})}
\frac{\smw_t(x_t^{\mathrm y})}{\mathcal V_s^\omega(\pi)}.
\]
The middle ratio is $1$ under exact matching.
\end{proof}

Exact matching turns the current term into a ratio of normalized welfare weights. The remaining term is the younger individual's within-person intertemporal tradeoff.

\subsection{Coarser and approximate matching}

A coarser rule may hold fixed only the state attribute and allow different values of $\contv_t$. The identity remains exact, and the error bound comes from the maintained bound on the difference in $\log \contv_t$.

\begin{proposition}[Private-value differences under coarser matching]\label{prop:coarse-matching}
Fix an ordered comparison $\bar\pi=(t,x_t^{\mathrm o},x_t^{\mathrm y},x_s^{\mathrm o})$ with $a(x_t^{\mathrm o})>a(x_t^{\mathrm y})$, $s=t+a(x_t^{\mathrm o})-a(x_t^{\mathrm y})$, $x_s^{\mathrm o}\in\Gamma_{t,s}(x_t^{\mathrm y})$, and $\Comp_t(x_t^{\mathrm o})=\Comp_t(x_t^{\mathrm y})$, without imposing equality of $\contv_t$. Then
\begin{equation}\label{eq:app-coarse-decomp}
\mathcal D^S(\bar\pi)=
\frac{\ndw_t(x_t^{\mathrm o})}{\ndw_t(x_t^{\mathrm y})}
\frac{\contv_t(x_t^{\mathrm o})}{\contv_t(x_t^{\mathrm y})}
\Rterm^S(\bar\pi),
\end{equation}
or equivalently
\begin{equation}\label{eq:app-coarse-identity}
\log\frac{\ndw_t(x_t^{\mathrm y})}{\ndw_t(x_t^{\mathrm o})}
=
\log \Rterm^S(\bar\pi)-\log\mathcal D^S(\bar\pi)
+\log\frac{\contv_t(x_t^{\mathrm o})}{\contv_t(x_t^{\mathrm y})}.
\end{equation}
\end{proposition}

\begin{corollary}[Approximate matching bound]\label{cor:approximate-matching}
Under \Cref{prop:coarse-matching}, if
\[
|\log\contv_t(x_t^{\mathrm o})-\log\contv_t(x_t^{\mathrm y})|\le\varepsilon,
\]
then
\begin{equation}\label{eq:app-approx-bound}
\left|
\log\frac{\ndw_t(x_t^{\mathrm y})}{\ndw_t(x_t^{\mathrm o})}
-\bigl(\log \Rterm^S(\bar\pi)-\log\mathcal D^S(\bar\pi)\bigr)
\right|\le\varepsilon.
\end{equation}
Equivalently,
\begin{align}
\log\frac{\ndw_t(x_t^{\mathrm y})}{\ndw_t(x_t^{\mathrm o})}
&\ge \log\Rterm^S(\bar\pi)-\log\mathcal D^S(\bar\pi)-\varepsilon,
\label{eq:app-approx-bound-lower}\\
\log\frac{\ndw_t(x_t^{\mathrm y})}{\ndw_t(x_t^{\mathrm o})}
&\le \log\Rterm^S(\bar\pi)-\log\mathcal D^S(\bar\pi)+\varepsilon.
\label{eq:app-approx-bound-upper}
\end{align}
When $\varepsilon=0$, the exact-matching identity obtains.
\end{corollary}

The sign-reversal construction is stated in \Cref{prop:coarse-ambiguity}. It follows directly from \Cref{eq:app-coarse-identity}: the private-value gap can dominate the discount and within-person terms.

Longer comparisons multiply along a specified chain. Appendix~\ref{app:technical-foundations} records the exact identity and distinguishes a fixed-chain calculation from the path-consistency requirement in \Cref{prop:imposed-discount-rationalization}.


\section{Valuing the Policy Margin}\label{sec:nonpaternalism}

This section derives the within-person term in the age comparison. The policy margin determines the future payoff, market compensation, and active constraint wedge. These objects must refer to the same marginal experiment and the same private unit.

Comparisons are per current living person. For one-period saving, the marginal experiment removes one current resource unit from the younger state and, conditional on survival, delivers a payoff over next-period states. Survival probabilities enter the payoff weights, $\Qterm$ prices the transfer, and $\WedgeTerm$ records the effect of active constraints. Alternative actuarial conventions define different policy margins and therefore different comparisons.

\subsection{Within-person valuation and matched states}

For state $(t,x)$, let $P_{t,x}$ be the benchmark continuation law. A \emph{within-person perturbation} on $M$ removes one current resource unit at $(t,x)$ and reallocates the compensating future payoff only across that individual's continuation states, holding benchmark prices, policies, and constraints fixed. Let $D^P_{t,x}$ and $D^{S,M}_{t,x}$ be the private and social derivatives on this policy margin.

\begin{assumption}[Consistent derivatives on the policy margin]\label{ass:represented-domain-consistency}
On $M$, every $(t,x)$ admits affine pathwise derivatives on admissible self-connected perturbations, with the null perturbation assigned value zero. When the same perturbation is described either pathwise or as the induced cross-sectional perturbation, the two derivatives agree on their common domain.
\end{assumption}

\begin{definition}[Nonpaternalism]\label{def:nonpaternalism}
A criterion is \emph{nonpaternalistic on $M$} if, for any within-person perturbations $\delta,\widetilde\delta$ at $(t,x)$, equality of private derivatives implies equality of social derivatives for the policy margin, and a strict private ranking implies the same strict social ranking. This is the baseline case $\BridgeCorr(\pi)=1$. Relative to the individual's ranking, a higher social value for the future payoff corresponds to $\BridgeCorr(\pi)<1$, and a lower social value corresponds to $\BridgeCorr(\pi)>1$.
\end{definition}

\subsection{Continuation adjustments}\label{subsec:bridge-assumptions}

The within-person valuation restriction states how the continuation payoff enters the planner's valuation. It is separate from the rule that matches current states. Formally, $\BridgeCorr:\Pi^M\to(0,\infty)$ is a measurable continuation adjustment on the declared comparison domain. It is fixed before the comparison network is constructed and cannot be chosen link by link to fit the welfare weights. Three cases determine its value.
\begin{enumerate}[label=(\Alph*), leftmargin=2.5em]
    \item \emph{Private-ranking benchmark.} If the planner uses the individual's ranking of the continuation payoff, then $\BridgeCorr(\pi)=1$ and $\Rterm^S(\pi)=\Qterm(\pi)\WedgeTerm(\pi)$.
    \item \emph{Paternalism about continuation.} If the planner values that continuation differently from the individual, then $\Rterm^S(\pi)=\Qterm(\pi)\WedgeTerm(\pi)\BridgeCorr(\pi)$. A higher social value corresponds to $\BridgeCorr(\pi)<1$ and a lower social value to $\BridgeCorr(\pi)>1$.
    \item \emph{Survival, health, or other continuation weights.} Ethical weights on survival, health, or future states enter $\BridgeCorr$, not $\Comp$ or $\contv$. They change the link value, but not current support unless they change the policy margin itself.
\end{enumerate}
Let $\CWedgeTerm(\pi):=\WedgeTerm(\pi)\BridgeCorr(\pi)$. Positivity makes $\log\BridgeCorr(\pi)$ well defined. For an exact edge,
\begin{equation}\label{eq:bridge-threshold}
\log\frac{\ndw_t(x_t^{\mathrm y})}{\ndw_t(x_t^{\mathrm o})}
=
\log\Qterm(\pi)-\log\mathcal D^S(\pi)+\log\WedgeTerm(\pi)+\log\BridgeCorr(\pi).
\end{equation}
Thus the normalized welfare weight is larger for the younger state whenever
\[
\BridgeCorr(\pi)>\frac{\mathcal D^S(\pi)}{\Qterm(\pi)\WedgeTerm(\pi)},
\]
with the corresponding $e^{\varepsilon}$ adjustment under approximate matching. The baseline sets $\BridgeCorr\equiv1$. Parametric continuation adjustments are maintained primitives. When they are bounded rather than fixed, the bounds enter through explicit restrictions on links or target contrasts. The quantitative application then shows how the identified ranking changes with the unit, constraint wedge, adjustment range, or support.

Appendix~\ref{app:technical-foundations} gives the corresponding verification result for a general convex household problem.

\begin{proposition}[Matched states for one-period saving]\label{cor:one-asset-map}
In the one-asset life-cycle problem of \Cref{sec:illustrations}, with state $x=(a,b,\varepsilon)$, age-dependent deterministic schedules, age-invariant transition matrix $P$, monotone private marginal value, and a one-period risk-free saving payoff, equality of $\contv$ pins down the scalar resource state within a transition-row class. The distribution of the future payoff is the survival-contingent row $P_{k\cdot}$, its market value is $1+r$, the relevant KKT term is the borrowing constraint on saving, and the support is the common one-period continuation support. Hence the non-age state information to hold fixed is the current earnings transition-row type,
\[
\Comp_a^M(b,\varepsilon_k)=P_{k\cdot}:=(P_{kk'})_{k'}.
\]
The matching rule also requires equality of $\contv_a$. Distinct rows identify current earnings states, coincident rows define row-equivalence classes, and any further coarsening requires a separate support check.
\end{proposition}

\begin{remark}[Verification in the one-asset model]\label{rem:worked-one-asset-map}
In the benchmark one-asset model, $\contv_a(b,\varepsilon_k)=u'(c_a(b,\varepsilon_k))$. Equality of $\contv_a$ recovers the scalar resource state inside a row class, while $P_{k\cdot}$ fixes the non-age continuation composition and support. The market and constraint terms are $\Qterm^M=1+r$ and $\WedgeTerm^M=[1-\lambda_a/\contv_a]^{-1}$. If regimes or multiple supporting prices change the payoff, price, support, or active constraint, the regime indicator enters $\Psi_a^M$.
\end{remark}

\subsection{Market valuation and constraints}

Fix an ordered comparison $\pi=(t,x_t^{\mathrm o},x_t^{\mathrm y},x_s^{\mathrm o})$, or a finite-support portfolio version of its future payoff. Let $\{\delta_R^\pi\}_{R\ge0}$ remove one current unit at $x_t^{\mathrm y}$ and add that future payoff at scale $R$.

\begin{definition}[Market compensation factor]\label{def:market-factor}
$\Qterm(\pi)$ is the gross future payoff that can be supported by one current unit along the policy margin. With a supporting state-price density, $\Qterm(\pi)$ is bounded by the inverse local price and equals it when the leg is locally purchasable.
\end{definition}

\begin{definition}[Private and social indifference factors]\label{def:represented-family}
When unique, $\widehat R^P(\pi)$ and $\widehat R^{S,M}(\pi)$ solve
\[
D^P_{t,x_t^{\mathrm y}}[\delta_{\widehat R^P(\pi)}^\pi]=0,
\qquad
D^{S,M}_{t,x_t^{\mathrm y}}[\delta_{\widehat R^{S,M}(\pi)}^\pi]=0.
\]
\end{definition}

\begin{proposition}[Social valuation of the policy margin]\label{prop:primitive-bridge}
If the social objective is differentiable and the marginal transfer has density $-1$ at the current state and density $R$ on the future payoff portfolio, then
\begin{equation}\label{eq:primitive-bridge-derivative}
D^{S,M}_{t,x_t^{\mathrm y}}[\delta_R^\pi]
=-\smw_t(x_t^{\mathrm y})+R\mathcal V_s^\omega(\pi),
\end{equation}
so the social indifference factor for the policy margin is
\begin{equation}\label{eq:primitive-bridge-factor}
\widehat R^{S,M}(\pi)=\Rterm^S(\pi)=\frac{\smw_t(x_t^{\mathrm y})}{\mathcal V_s^\omega(\pi)}.
\end{equation}
\end{proposition}

\begin{assumption}[Regularity of the marginal transfer]\label{ass:represented-family-bridge}
For every exact, approximate, or coarse comparison used below, either \Cref{prop:primitive-bridge} applies or the marginal transfer is restricted to a domain satisfying \Cref{eq:primitive-bridge-derivative,eq:primitive-bridge-factor}.
\end{assumption}

\begin{proposition}[Market compensation under constraints]\label{prop:supporting-prices-broad}
Suppose $\Qterm(\pi)$ is the local marginal exchange rate for the future payoff and the sign-normalized KKT loading of that transfer is $\Xi(\pi)\in[0,\contv_t(x_t^{\mathrm y}))$. Then
\begin{equation}\label{eq:private-comp-broad}
\widehat R^P(\pi)=\Qterm(\pi)\left[1-\frac{\Xi(\pi)}{\contv_t(x_t^{\mathrm y})}\right]^{-1}.
\end{equation}
The bracketed term is one when all relevant loadings vanish.
\end{proposition}

\begin{corollary}[Convex liquid/illiquid extension]\label{cor:liquid-illiquid-bridge}
In a convex liquid/illiquid extension with verified supporting prices, the same KKT formula applies after the comparison records the current illiquid state and relevant adjustment-cost gradients.
\end{corollary}

\begin{remark}[Outside the convex class]\label{rem:nonconvex-scope}
Outside the convex class, supporting prices may fail to be unique and loadings may vary with the active regime. The regime indicators needed to identify the payoff, price, or active constraint then enter the verified comparison.
\end{remark}

\begin{remark}[Verification checklist for matched states]\label{rem:verification-recipe}
For policy margin $M$, specify the future payoff, derive its market value, list the active constraints, collect the nonresource state information in $\Psi_a^M$, and verify on support that equality of $(\Psi_a^M,\contv_a)$ holds all of these objects fixed.
\end{remark}

\begin{assumption}[Local supporting price representation]\label{ass:supporting-prices}
Every comparison using market valuation has a verified supporting-price/KKT argument yielding $\Xi(\pi)\in[0,\contv_t(x_t^{\mathrm y}))$ such that
\begin{equation}\label{eq:private-comp-representation}
\widehat R^P(\pi)=\Qterm(\pi)\left[1-\frac{\Xi(\pi)}{\contv_t(x_t^{\mathrm y})}\right]^{-1}.
\end{equation}
If the perturbation is locally spanned and all relevant constraints are slack, then $\Xi(\pi)=0$.
\end{assumption}

\begin{proposition}[Within-person valuation under constraints]\label{prop:implementability}
Suppose \Cref{ass:benchmark,ass:localized-cell,ass:local-derivatives,ass:supporting-prices,ass:represented-domain-consistency,ass:represented-family-bridge} hold and the criterion is nonpaternalistic on $M$. Then every verified comparison satisfies
\begin{equation}\label{eq:R-market-wedge}
\Rterm^S(\pi)=\Qterm(\pi)\WedgeTerm(\pi),
\qquad
\WedgeTerm(\pi)=\left[1-\frac{\Xi(\pi)}{\contv_t(x_t^{\mathrm y})}\right]^{-1}\ge1,
\end{equation}
so $\Rterm^S(\pi)\ge\Qterm(\pi)$ in the baseline $\BridgeCorr\equiv1$ case. Equality holds whenever $\Xi(\pi)=0$; a sufficient primitive condition is local spanning and local slackness. With a correction to the within-person valuation restriction, define $\CWedgeTerm(\pi)=\WedgeTerm(\pi)\BridgeCorr(\pi)$ and replace the right side by $\Qterm(\pi)\CWedgeTerm(\pi)$.
\end{proposition}

\begin{remark}[Economic meaning of the wedge]\label{rem:wedge-economic}
$\WedgeTerm(\pi)$ is the inverse share of the current private marginal value available for the future payoff after the KKT terms are netted out. If $\mathcal D^S(\pi)=\Qterm(\pi)$, then \Cref{thm:mr,prop:implementability} give $\ndw_t(x_t^{\mathrm y})/\ndw_t(x_t^{\mathrm o})=\WedgeTerm(\pi)$.
\end{remark}

\begin{proposition}[Target comparisons in private-value units]\label{prop:policy-margin-reduction}
Fix a financing rule for the policy margin and write two local policy perturbations in $\PVEU$ units as $\Delta^A=\tau^A-\phi$ and $\Delta^B=\tau^B-\phi$, with common financing measure $\phi$. If both are admissible and $\PVEU$-integrable, then
\begin{equation}\label{eq:policy-margin-reduction}
D\SWF[\Delta^A]-D\SWF[\Delta^B]
=\int \ndw_t(x)(\tau^A-\tau^B)(\dd x).
\end{equation}
If $\tau^A(\Pop)=\tau^B(\Pop)>0$, the sign equals the sign of the difference in $\ndw$-weighted target averages.
\end{proposition}

\begin{remark}[Policy to margin mapping]\label{rem:policy-mapping}
A full policy maps into this local object only after specifying the financing rule, private unit, support, and component normalizations. The quantitative section reports the supported local target object and explicit omitted-support layers.
\end{remark}

\begin{corollary}[State-price form of the restriction]\label{cor:market-value}
If the future payoff is locally purchasable at state-price density $q_{t,s}$, then $\Rterm^S(\pi)=q_{t,s}(x_s^{\mathrm o}\mid x_t^{\mathrm y})^{-1}\WedgeTerm(\pi)$ in the singleton case, with the analogous portfolio price for finite-support legs.
\end{corollary}


\section{Coherence and Identification}\label{sec:age-blindness}

The local decomposition yields two identification results. First, social discounting below market compensation implies a larger normalized welfare weight on the younger state when the comparison is exact and the maintained within-person restriction holds. Second, the local ratios identify one coherent welfare-weight system only when the consistency conditions in \Cref{prop:imposed-discount-rationalization} hold. The network representation follows the integrability logic familiar from revealed preference and convex analysis \citep{Afriat1967,Rockafellar1970,Varian1982}, while the economic environment determines which states can be linked and how each link is valued.

\begin{proposition}[Age ranking under exact matching]\label{thm:impossibility}
Let $\pi=(t,x_t^{\mathrm o},x_t^{\mathrm y},x_s^{\mathrm o})$ be an exactly matched comparison for which the future payoff, supporting-price/KKT representation, and within-person valuation restriction are verified. Suppose \Cref{ass:benchmark,ass:localized-cell,ass:local-derivatives,ass:supporting-prices,ass:represented-domain-consistency,ass:represented-family-bridge} hold and the criterion is nonpaternalistic on the policy margin $M$. Then
\begin{equation}\label{eq:impossibility-ratio}
\log\frac{\ndw_t(x_t^{\mathrm y})}{\ndw_t(x_t^{\mathrm o})}
=
\log\Qterm(\pi)-\log\mathcal D^S(\pi)+\log\WedgeTerm(\pi)
\ge
\log\Qterm(\pi)-\log\mathcal D^S(\pi).
\end{equation}
Thus, if
\begin{equation}\label{eq:low-social-discounting}
\mathcal D^S(\pi)<\Qterm(\pi),
\end{equation}
then
\begin{equation}\label{eq:young-tilt}
\ndw_t(x_t^{\mathrm y})>\ndw_t(x_t^{\mathrm o}).
\end{equation}
Conditional age neutrality in normalized welfare weights fails on any matched set that supports such a comparison. If $\mathcal D^S(\pi)=\Qterm(\pi)$, then
\[
\frac{\ndw_t(x_t^{\mathrm y})}{\ndw_t(x_t^{\mathrm o})}=\WedgeTerm(\pi),
\]
so any remaining younger tilt is entirely due to the constraint wedge.
\end{proposition}

\begin{remark}[Scope of the nonpaternalistic benchmark]\label{rem:nonpaternalism-sensitivity}
With a continuation adjustment, set $\CWedgeTerm(\pi):=\WedgeTerm(\pi)\BridgeCorr(\pi)$. Then \Cref{eq:impossibility-ratio} becomes $\log(\ndw_t(x_t^{\mathrm y})/\ndw_t(x_t^{\mathrm o}))=\log\Qterm(\pi)-\log\mathcal D^S(\pi)+\log\CWedgeTerm(\pi)$, and the younger tilt conclusion requires $\mathcal D^S(\pi)<\Qterm(\pi)\CWedgeTerm(\pi)$. The baseline maintained in the numerical diagnostic is $\BridgeCorr\equiv1$.
\end{remark}

\subsection{Supported comparison networks}

\begin{definition}[Supported comparison network]\label{def:exact-comparison-graph}
Fix date $t$ and a policy margin. The supported comparison network $G_t=(V_t,E_t)$ has vertices given by current states that belong to a populated matched set. A directed link $e=(x^{\mathrm y},x^{\mathrm o})$ belongs to $E_t$ when $a(x^{\mathrm o})>a(x^{\mathrm y})$, the two states are exactly matched, and some continuation state $x_s^{\mathrm o}\in\Gamma_{t,s}(x^{\mathrm y})$ reaches age $a(x^{\mathrm o})$.
\end{definition}

The network is defined on current states. Future payoffs determine the link labels. Exact matches are the conceptual benchmark. Finite applications use explicit tolerance bands for private marginal values.

\begin{proposition}[Local welfare-ratio labels]\label{prop:exact-edge-weight}
For a link $e=(x^{\mathrm y},x^{\mathrm o})$ and continuation state $x_s^{\mathrm o}$, write
\[
\pi(e;x_s^{\mathrm o}):=(t,x^{\mathrm o},x^{\mathrm y},x_s^{\mathrm o})
\]
and define the corrected constraint term $\CWedgeTerm(\pi):=\WedgeTerm(\pi)\BridgeCorr(\pi)$, with $\BridgeCorr\equiv1$ in the baseline. The edge label is
\begin{equation}\label{eq:edge-weight-witness}
g_t(e;x_s^{\mathrm o})
:=
\log\Qterm(\pi(e;x_s^{\mathrm o}))
-
\log\mathcal D^S(\pi(e;x_s^{\mathrm o}))
+
\log\CWedgeTerm(\pi(e;x_s^{\mathrm o})).
\end{equation}
Under the assumptions of \Cref{thm:impossibility} in the baseline, and under those assumptions plus the continuation-adjusted relation $\Rterm^S=\Qterm\WedgeTerm\BridgeCorr$ when $\BridgeCorr\neq1$,
\begin{equation}\label{eq:edge-weight-gradient}
g_t(e;x_s^{\mathrm o})=
\log\ndw_t(x^{\mathrm y})-
\log\ndw_t(x^{\mathrm o}).
\end{equation}
Hence the label is independent of the continuation state used to verify the link. Write it as $g_t(e)$.
\end{proposition}

The coherence theorem is stated in \Cref{prop:imposed-discount-rationalization}. The construction in this section supplies its edge labels and supported comparison network.

\begin{proposition}[Path consistency of the implied welfare ratios]\label{thm:path-independence}
For a directed exact edge chain $c=(x^0,x^1,\ldots,x^n)$ in $G_t$, define
\begin{equation}\label{eq:path-total-weight}
G_t(c):=\sum_{\ell=1}^n g_t(x^{\ell-1},x^\ell).
\end{equation}
Under the assumptions of \Cref{thm:impossibility},
\begin{equation}\label{eq:path-potential-identity}
G_t(c)=\log\ndw_t(x^0)-\log\ndw_t(x^n).
\end{equation}
Consequently: (i) two directed exact edge chains with the same endpoints have the same total weight; (ii) every closed walk in the underlying undirected network has zero signed total weight; and (iii) on each connected component $C$ of that network there is a function $\phi_t^C:C\to\RR$, unique up to an additive constant, such that
\[
g_t(x^{\mathrm y},x^{\mathrm o})=\phi_t^C(x^{\mathrm y})-\phi_t^C(x^{\mathrm o})
\]
for every directed edge in $C$. One choice is $\phi_t^C=\log\ndw_t|_C$. Thus the function is the log normalized welfare weight, up to the component constant.
\end{proposition}

\begin{corollary}[Component normalization and recursion]\label{cor:component-normalization}
Fix a connected component $C$ and a normalization state $x^\star\in C$. Once $\ndw_t(x^\star)>0$ is fixed, every within component ratio is uniquely determined by exact edge weights. Any exact edge spanning tree enforcing
\[
\ndw_t(x^{\mathrm y})=\exp(g_t(x^{\mathrm y},x^{\mathrm o}))\ndw_t(x^{\mathrm o})
\]
on tree edges gives the same ratios. Multiplying the normalization level by a common positive constant rescales the component (and changes no ratio).
\end{corollary}

\begin{remark}[Outer envelope across components]\label{rem:outer-component-envelope}
Disconnected components leave relative constants unidentified. A conservative closure in current goods sets
$\bar m:=\log(\max_{x\in S}\contv_t(x)/\min_{x\in S}\contv_t(x))$, the log range of private marginal values on the maintained support. This does not identify relative levels across components. It only supplies an outer bound on the missing exchange rate. If $[\underline\Delta,\overline\Delta]$ is an interval under a fixed component normalization, then $[\underline\Delta-\bar m,\overline\Delta+\bar m]$ is the component-envelope interval.
\end{remark}

\begin{remark}[What the comparison network adds]\label{rem:graph-adds}
The network turns pairwise restrictions into recursions within connected components and makes support gaps and normalization dependence explicit. Nonzero cycle residuals diagnose edge labels that are imposed or operational, for example labels generated by a finite grid, fixed nearest-neighbor correspondences, or external schedules. Exact endogenous labels telescope as differences in log normalized welfare weights. The restriction is local to the comparison network at date $t$. Consistency of welfare weights across dates would require an enlarged network across dates that links each future payoff dated $s$ to its appearance as a current state at date $s$.
\end{remark}

\begin{corollary}[Approximate pairwise band]\label{cor:approx-impossibility}
Let $\bar\pi=(t,x_t^{\mathrm o},x_t^{\mathrm y},x_s^{\mathrm o})$ be a verified approximate comparison in the sense of \Cref{def:ordered-bridge-tuple}, with $s=t+a(x_t^{\mathrm o})-a(x_t^{\mathrm y})$, $a(x_t^{\mathrm o})>a(x_t^{\mathrm y})$, $\Comp_t(x_t^{\mathrm o})=\Comp_t(x_t^{\mathrm y})$, $x_s^{\mathrm o}\in\Gamma_{t,s}(x_t^{\mathrm y})$, and
\[
|\log\contv_t(x_t^{\mathrm o})-\log\contv_t(x_t^{\mathrm y})|\le\varepsilon.
\]
Under the assumptions of \Cref{thm:impossibility},
\begin{equation}\label{eq:approx-impossibility-band}
\left|
\log\frac{\ndw_t(x_t^{\mathrm y})}{\ndw_t(x_t^{\mathrm o})}
-
\bigl[
\log\Qterm(\bar\pi)-\log\mathcal D^S(\bar\pi)+\log\CWedgeTerm(\bar\pi)
\bigr]
\right|\le\varepsilon.
\end{equation}
Equivalently,
\begin{equation}\label{eq:approx-impossibility-ratio}
\log\frac{\ndw_t(x_t^{\mathrm y})}{\ndw_t(x_t^{\mathrm o})}
\in
\left[
\log\Qterm(\bar\pi)-\log\mathcal D^S(\bar\pi)+\log\CWedgeTerm(\bar\pi)-\varepsilon,
\ \log\Qterm(\bar\pi)-\log\mathcal D^S(\bar\pi)+\log\CWedgeTerm(\bar\pi)+\varepsilon
\right].
\end{equation}
If
\begin{equation}\label{eq:approx-impossibility-condition}
\mathcal D^S(\bar\pi)<\Qterm(\bar\pi)\CWedgeTerm(\bar\pi)e^{-\varepsilon},
\end{equation}
then $\ndw_t(x_t^{\mathrm y})>\ndw_t(x_t^{\mathrm o})$. In the baseline $\BridgeCorr\equiv1$, $\WedgeTerm\ge1$, so the wedge-free sufficient condition $\mathcal D^S(\bar\pi)<\Qterm(\bar\pi)e^{-\varepsilon}$ also implies a larger normalized weight for the younger state. The opposite sign requires $\mathcal D^S(\bar\pi)>\Qterm(\bar\pi)\CWedgeTerm(\bar\pi)e^\varepsilon$.
\end{corollary}

Appendix~\ref{app:technical-foundations} gives the corresponding chain bound, finite-grid convergence result, and regularity conditions for nearby states.

The restriction is local to the policy margin and the current comparison network. It identifies differences in normalized welfare weights under the maintained state match, private unit, and within-person valuation restriction. Global policy rankings, consistency across dates, independent age ethics, and criteria that do not respect private rankings require additional structure.

The exact matching implication is stated in \Cref{cor:trilemma}. Appendix~\ref{app:technical-foundations} defines the comparison-system measure used for its almost-everywhere statement.


\section{Quantitative Application}\label{sec:illustrations}

This section applies the theory in a calibrated life-cycle economy. The maintained objects are the one-period saving margin $M$, the state-matching map $\Comp$, and the current-goods private unit $\contv=u'(c)$. Each result is conditional on those objects unless a row explicitly recomputes the comparison network. The tables report which age rankings are identified under alternative support, unit, wedge, and component-normalization assumptions.

\subsection{One-asset benchmark and market valuation}\label{subsec:two-state-benchmark}

The analytical benchmark is a finite-life one-asset economy with uninsurable labor income risk and a borrowing limit, following, inter alia, \citet{Huggett1993}, \citet{Aiyagari1994}, and \citet{ConesaKrueger1999}. With $x=(a,b,\varepsilon)$, households solve
\begin{equation}\label{eq:section9-one-asset-value}
V_a(b,\varepsilon;m)
=
\max_{b'\ge\underline b_a}
\left\{u(c+m)+\beta p_a\,\EE[V_{a+1}(b',\varepsilon';0)\mid\varepsilon]\right\},
\qquad
c+b'=y_a(\varepsilon)+(1+r)b,
\end{equation}
with terminal condition $V_A(b,\varepsilon;m)=u(c_A+m)$. In terms of the calibration, I use \citet{ConesaKrueger1999} for demographics and earnings, \citet{GourinchasParker2002} for preferences and the interest rate, and \citet{NCHS2024LifeTable2021} for survival: entry age $20$, retirement age $65$, terminal age $85$, population growth $1.1$ percent, $\beta=0.9569$, $\sigma=1.3969$, $r=0.0344$, replacement ratio $0.5$, no unsecured debt, an $800$-point curved asset grid, and earnings states $(0.73,1.27)$ with rows $(0.82,0.18)$ and $(0.18,0.82)$. The policy margin is one-period survival-contingent risk-free saving.

\begin{proposition}[Risk-free saving specialization and KKT-consistent Euler wedge]\label{prop:constraint-amplification}
In \Cref{eq:section9-one-asset-value}, fix a younger state $(a,b_i,\varepsilon_k)$, let $g_a(i,k)$ be the chosen next asset index, and let $m_a(i,k)$ be the direct $\log\contv$ matched older asset index. On the one-step risk-free portfolio that removes one current unit and adds the same next-period unit in each survival-contingent continuation earnings state at asset $b_{g_a(i,k)}$, \Cref{ass:supporting-prices} holds with
\begin{equation}\label{eq:section9-riskfree-Q}
\Qterm_a^{(i,k)}=1+r
\end{equation}
and
\begin{equation}\label{eq:section9-riskfree-Lambda}
\WedgeTerm_a^{(i,k)}
=
\frac{u'(c_a(b_i,\varepsilon_k))}{\beta p_a(1+r)\sum_{k'}P_{kk'}u'(c_{a+1}(b_{g_a(i,k)},\varepsilon_{k'}))}
\ge1.
\end{equation}
For the one-period comparison $\pi_{aik}$ used by the link,
\begin{equation}\label{eq:section9-riskfree-bound}
\frac{\ndw_a(b_i,\varepsilon_k)}{\ndw_{a+1}(b_{m_a(i,k)},\varepsilon_k)}
=
\frac{1+r}{\mathcal D^S(\pi_{aik})}\WedgeTerm_a^{(i,k)}
\frac{\contv_{a+1}(b_{m_a(i,k)},\varepsilon_k)}{\contv_a(b_i,\varepsilon_k)}.
\end{equation}
The final ratio equals one under exact matching and lies in $[e^{-\varepsilon},e^{\varepsilon}]$ under tolerance radius $\varepsilon$. Replacing the endogenous comparison ratios by an imposed age schedule $\bar{\mathcal D}^S_a$ turns \Cref{eq:section9-riskfree-bound} into the network restriction tested below.
\end{proposition}

For a given comparison, the endogenous $\mathcal D^S(\pi)$ is an identity. An imposed schedule
$\bar{\mathcal D}^S_a$ is a restriction. Unless stated otherwise, the diagnostics use the market
benchmark $\bar\rho^S=r$, and the future payoff is risk free, spans one period, and is conditional
on survival. The verified current map is the transition row, which in this benchmark coincides
with the current earnings state. The symbol $g_a(i,k)$ defines the future payoff
from the grid policy, and the maintained wedge is the KKT-consistent projection rather than the
raw Euler residual on the finite grid. Overlap for this leg is essentially complete in the benchmark
with two earnings states. Overall overlap and the minimum overlap by age both equal
$\BaselineOverlapOverall$, and the mean log mismatch for the young, weighted by mass, is
$\BaselineMatchMean$. The current map is verified by solving the household problem, computing
$\contv=u'(c)$, recording the transition row and KKT loading for the saving margin,
and reporting overlap, mismatch, and dispersion diagnostics before target statistics are computed.
In the tolerance network with five earnings states, $\max|\Delta\log Q|=\MarginAuditMaxDeltaLogQ$,
the transition row mismatch share is $\MarginAuditRowMismatchShare$, and the mean and 95th
percentile of $|\Delta\log\Lambda|$ are $\MarginAuditMeanDeltaLogLambda$ and
$\MarginAuditPninetyFiveDeltaLogLambda$. The replication output also records flags for boundary
regimes and adjacent monotonicity.

\subsection{Five-state calibration and identified set}\label{subsec:graph-diagnostic}

The main calibration uses the five-state hump block. It keeps preferences, ages, survival, borrowing limit, and interest rate, uses a $450$-point asset grid, adds hump-shaped deterministic earnings, and uses a five-state Rouwenhorst process with levels $(0.564,0.738,0.965,1.261,1.649)$ and persistence $0.82$. The focal comparison uses the central current-earnings state in this calibration. A warm-up statistic with target $A$ at ages $20$ to $29$ and target $B$ at ages $40$ to $49$ in the central current-earnings row has verified-map overlap $0.7145$ and interval $[0.00090,0.01444]$ under the stated tolerance protocol. Age matching raises overlap to $0.9673$ and widens the interval to $[-0.02605,1.19061]$. The additional overlap therefore comes at the cost of larger private-value differences and weaker identifying restrictions.

Let $I$ be the retained set of current states on the supported comparison domain. Target weights may be zero on connector states used only to impose pairwise restrictions. Let $E_{\varepsilon}$ be the maintained set of tolerance links. A link is retained when the non-age code matches and the private-unit mismatch is within its stated tolerance. The search cap $K$ is a computational cap used to find all retained candidates inside the tolerance set. For link $e=(y,o)$, let $H_e:=\log\Qterm(e)-\log\bar{\mathcal D}^S(e)+\log\CWedgeTerm(e)$ and mismatch radius $\varepsilon_e$. The baseline sets $\BridgeCorr\equiv1$, so $\CWedgeTerm=\WedgeTerm$; Table~\ref{tab:unit-wedge-protocol} also reports a bounded continuation-adjustment range. With one normalization state $i_C$ per connected component, feasible log welfare weights are
\begin{equation}\label{eq:finite-graph-program}
\mathcal Z_{\varepsilon}:=\{z\in\RR^I: z_{i_C}=0\ \text{for all }C,\quad H_e-\varepsilon_e\le z_y-z_o\le H_e+\varepsilon_e\ \text{for all }e\in E_{\varepsilon}\}.
\end{equation}
For target mass vectors $w^A,w^B$ on the supported domain,
\begin{equation}\label{eq:finite-graph-objective}
\Delta_{AB}(z):=
\log\frac{\sum_i w_i^A e^{z_i}}{\sum_i w_i^A}
-
\log\frac{\sum_i w_i^B e^{z_i}}{\sum_i w_i^B},
\qquad
[\underline\Delta_{AB},\overline\Delta_{AB}]
=
\left[\inf_{z\in\mathcal Z_{\varepsilon}}\Delta_{AB}(z),\sup_{z\in\mathcal Z_{\varepsilon}}\Delta_{AB}(z)\right].
\end{equation}
The interval is the identified set under the maintained tolerance. Point estimates are secondary summaries. Nonzero cycle or duplicate-link residuals reject coherence of the imposed local schedule.

If component constants $c_C$ are left free after solving within-component potentials $\tilde z$, then
\begin{equation}\label{eq:component-anchor-sensitivity}
\Delta_{AB}(c)=
\log\frac{\sum_C e^{c_C}A_C}{\sum_iw_i^A}
-
\log\frac{\sum_C e^{c_C}B_C}{\sum_iw_i^B},
\quad
A_C=\sum_{i\in C}w_i^Ae^{\tilde z_i},\quad
B_C=\sum_{i\in C}w_i^Be^{\tilde z_i}.
\end{equation}
Targets that load differently across disconnected components require a component-exposure report. A bounded-normalization sensitivity set restricts $c_C$ to $[-\bar c,\bar c]$. Increasing $\bar c$ traces the minimum cross-component restriction needed for sign identification. In the focal comparison, the sign is lost once $\bar c=0.00333$ log points, about $0.33$ percent log points. The unrestricted-normalization row is the limiting case and is unidentified when unrounded target mass lies in components not shared by the other target.

The headline comparison appears in \Cref{tab:main-diagnostic-audit}. Appendix~\ref{app:additional-quantitative} reports the full unit, support, wedge, and richer-state results.

The focal comparison has no unconditional sign. The fixed-normalization interval is negative, the largest shared component gives a positive interval, and unrestricted component scales leave the aggregate contrast unidentified. The difference comes from asymmetric exposure to disconnected components, not from disagreement among local comparisons.

\begin{table}[htbp]
\centering
\caption{Discount schedules and identified age rankings}
\label{tab:discount-schedule-frontier}
\footnotesize
\renewcommand{\arraystretch}{1.05}
\resizebox{0.99\textwidth}{!}{%
\begin{tabular}{lccccc}
\toprule
Candidate annual $\rho^S$ & Fixed-normalization interval & Unrestricted status & Largest-component interval & Conditional sign & Unrestricted sign? \\
\midrule
0.00\% & $[0.65797,\ 0.66309]$ & Unidentified & $[0.66200,\ 0.66254]$ & positive & No \\
2.00\% & $[0.26804,\ 0.27311]$ & Unidentified & $[0.27570,\ 0.27621]$ & positive & No \\
3.00\% & $[0.07594,\ 0.08098]$ & Unidentified & $[0.08540,\ 0.08588]$ & positive & No \\
3.44\% & $[-0.00819,\ -0.00298]$ & Unidentified & $[0.00004,\ 0.00082]$ & negative & No \\
4.00\% & $[-0.11429,\ -0.10928]$ & Unidentified & $[-0.10307,\ -0.10261]$ & negative & No \\
\bottomrule
\end{tabular}%
}
\begin{minipage}{0.99\textwidth}
\footnotesize Notes. The table varies the imposed local social discount schedule while holding the calibration, target pair, policy margin, and current-goods private unit fixed. The market benchmark is $\rho^S=r=3.44\%$. With unrestricted component normalizations, every aggregate sign is unidentified.
\end{minipage}
\end{table}


\subsection{Local target comparison on common support}\label{subsec:nearby-target-pair}

Incomplete overlap separates the comparison identified on common support from a full policy ranking in physical resource units. The target pair asks whether an age ranking survives when young and old states share the same policy-relevant information and private unit.

\begin{proposition}[What common-support target comparisons identify]\label{prop:common-support-selection}
Fix two local perturbations measured in private-value units satisfying \Cref{prop:policy-margin-reduction}, and let $S\subseteq\Pop$ be any maintained support set. Then
\begin{equation}\label{eq:common-support-decomposition}
D\SWF[\Delta^A]-D\SWF[\Delta^B]
=
\int_S\ndw_t(x)(\tau^A-\tau^B)(\dd x)
+
\int_{S^c}\ndw_t(x)(\tau^A-\tau^B)(\dd x).
\end{equation}
A common-support comparison identifies the first term conditional on the maintained welfare-weight system, including the component normalizations and any point-summary rule. The reported $\Delta_{AB}$ is a supported contrast in normalized welfare weights on $S$; when retained masses differ, it need not equal the unnormalized level contribution in sign or scale. If $\ndw_t(x)\in[\underline\eta,\overline\eta]$ on $S^c$ for $(\tau^A+\tau^B)$ almost everywhere $x$, the omitted term lies in
\begin{equation}\label{eq:common-support-interval}
\left[
\underline\eta\,\tau^A(S^c)-\overline\eta\,\tau^B(S^c),
\ \overline\eta\,\tau^A(S^c)-\underline\eta\,\tau^B(S^c)
\right].
\end{equation}
\end{proposition}

Support and within-support rematching both matter. Target $A$ puts equal mass on ages $25$ to $34$ in the central current earnings row. Target $B$ puts equal mass on ages $45$ to $54$ on the same policy margin. Age matching gives $\Delta_{AB}=0.01292$ on full support and $0.00697$ on common support. On that same domain and under fixed component normalizations, normalized welfare weights from matched states give $\Delta_{AB}=-0.00481$ with identified interval $[-0.00819,-0.00298]$. The $-0.01773$ log-point shift from age-based full support to verified current-state common support decomposes into $-0.00595$ from support exclusion and $-0.01178$ from within-support rematching. Common support retains $0.3706$ of target $A$ mass and $0.9338$ of target $B$ mass.

Disconnected support is the central identification issue. The maintained comparison network has $72{,}171$ vertices, $249{,}501$ edges, $3{,}737$ components, and $181{,}067$ independent cycles; target mass largely lies in cyclic components. Target $A$ is concentrated in one component, target $B$ is dispersed, and unrestricted component normalizations make the comparison unbounded. The current-goods outer envelope in \Cref{rem:outer-component-envelope} uses global private marginal-utility log range $2.16058$ and gives $[-2.16877,2.15760]$: finite, explicit, and too wide to identify a sign. The fixed-normalization and largest-component rows are conditional comparisons rather than cross-component identification.

\begin{table}[htbp]
\centering
\caption{Component exposure for the focal target pair}
\label{tab:component-exposure}
\footnotesize
\renewcommand{\arraystretch}{1.05}
\begin{tabular}{lcc}
\toprule
Component exposure statistic & Target $A$ & Target $B$ \\
\midrule
Mass in components shared with the other target & 1.0000 & 0.2797 \\
Component Herfindahl index & 0.9999 & 0.1742 \\
Mass in largest common component & 0.3706 & 0.2612 \\
\bottomrule
\end{tabular}
\begin{minipage}{0.88\textwidth}
\footnotesize Notes. The first row is normalized by each target's network-supported mass; the Herfindahl index uses retained component shares. The largest-common-component row is normalized by original target mass. The unrounded target-$A$ shared share is $0.99996$, so a small $A$-only mass remains.
\end{minipage}
\end{table}


\subsection{Robustness and scope}

The additional exercises target the main threats to the interpretation. Rebuilding the matched-state network under alternative private units preserves the fixed-normalization sign for moderate rescalings but removes it under stronger age-state rescaling. Propagating raw finite-grid Euler residuals as wedges also removes sign identification, while the maintained KKT projection does not. Grid refinement leaves the central objects stable, and the finite program has zero parallel-label, cycle, band, and relaxation residuals at the reported precision. A borrowing-bound support check moves the focal interval across zero, showing that the result is sensitive to the supported estimand rather than to numerical feasibility. A richer liquid-illiquid state space likewise yields intervals containing zero. Appendix~\ref{app:additional-quantitative} reports the full diagnostics and tables.


\section{Conclusion}\label{sec:conclusion}

A social discount rate does not by itself rank young and old agents in a heterogeneous economy. The ranking also depends on the current states being compared, the private unit used across them, the support connecting those states, and the normalization across disconnected components.

The paper gives exact conditions for an age-priority interpretation. Matching policy-relevant state information and private marginal value removes state and unit differences from the local comparison. Path consistency then determines whether the supported local ratios come from one coherent system of normalized welfare weights. Disconnected support identifies that system only within components.

The calibrated life-cycle application shows that these conditions alter the economic conclusion. Age-only comparisons are uninformative. Matched local comparisons can favor younger agents while the aggregate young-old ranking is negative under one component normalization, positive on the largest shared component, and unidentified without a cross-component normalization. Claims about age priority therefore require more than a social discount rate. They require a stated welfare object, connected support, and enough normalization to identify the comparison.


\appendix
\numberwithin{equation}{section}
\numberwithin{table}{section}
\numberwithin{figure}{section}

\section{Technical Foundations and Extensions}\label{app:technical-foundations}

This appendix collects definitions and extensions that support the main argument but are not needed to read the economic results.

\subsection{Local evaluation on finite and continuous supports}
\begin{assumption}[Localized cell evaluation]\label{ass:localized-cell}
For finite grids, a displayed shift at cell $x$ is the density $h_t(x)$ in
$\Delta_t(\dd x)=h_t(x)\dist_t(\dd x)$. The aggregate resource change associated
with that cell is therefore $h_t(x)\dist_t(\{x\})$, with $\dist_t$ held fixed. On
continuous supports, fix the Euclidean differentiation basis on the Borel state
coordinates, or another stated standard differentiation basis. A displayed shift at
$x$ is the limit, when it exists, of admissible perturbations with uniformly bounded
densities, supported on neighborhoods shrinking to $x$ under the chosen basis, after
division by the $\dist_t$-measure of the neighborhood. Thus the object is the local
density $\dd\Delta_t/\dd\dist_t$ evaluated at $x$. Identities involving
$(\smw_t,\contv_t,\ndw_t)$ are asserted after fixing versions of these objects on
sets of full measure. Absent such versions, they are understood as identities in the
relevant conditional essential range.
\end{assumption}

\subsection{Private-value units}

\begin{definition}[Private value equivalent unit]\label{def:pcu}
When $0<\contv_t(x)<\infty$, one \emph{private-value-equivalent unit}, or $\PVEU$, at $(t,x)$ is the local resource shift of size $1/\contv_t(x)$.
\end{definition}

A $\PVEU$ raises private continuation value by one marginal unit. An admissible $\Delta$ is $\PVEU$-integrable when $\sum_t\int\contv_t\,\dd|\Delta_t|<\infty$ over active dates.

\begin{proposition}[PVEU representation]\label{prop:pcu-representation}
For any $\PVEU$-integrable admissible $\Delta$, define
\[
\Delta_t^{\PVEU}(\dd x):=\contv_t(x)\Delta_t(\dd x).
\]
Then
\begin{equation}\label{eq:pcu-representation}
D\SWF[\Delta]
= \sum_{t\ge0}\int_{\Pop}\smw_t(x)\Delta_t(\dd x)
= \sum_{t\ge0}\int_{\Pop}\ndw_t(x)\Delta_t^{\PVEU}(\dd x).
\end{equation}
Thus $\ndw_t(x)$ is the social value of one $\PVEU$ delivered to cell $x$ on this integrable local domain.
\end{proposition}

\begin{proof}
Substitute $\smw_t=\ndw_t\contv_t$ into \Cref{eq:smw} and use the definition of $\Delta_t^{\PVEU}$.
\end{proof}

For cells with positive private marginal values,
\[
\frac{\smw_t(x)}{\smw_t(x')}=\frac{\contv_t(x)}{\contv_t(x')}\frac{\ndw_t(x)}{\ndw_t(x')}.
\]
Weights in private-value units allow unequal private marginal values. Exact matching is needed to interpret a resource-unit comparison as age-related social priority.

\subsection{Verification certificates for policy margins}

A verification certificate lists the future payoff induced by the policy margin, its market price, and the relevant constraint loadings. The matched-state code retains the non-age coordinates needed to identify those objects once age alignment and the private marginal value are fixed.

\begin{table}[htbp]
\centering
\caption{Verification certificate for matched states}
\label{tab:primitive-certificate}
\scriptsize
\renewcommand{\arraystretch}{1.05}
\resizebox{0.99\textwidth}{!}{%
\begin{tabular}{>{\raggedright\arraybackslash}p{0.19\textwidth}>{\raggedright\arraybackslash}p{0.22\textwidth}>{\raggedright\arraybackslash}p{0.18\textwidth}>{\raggedright\arraybackslash}p{0.18\textwidth}>{\raggedright\arraybackslash}p{0.24\textwidth}}
\toprule
Policy margin & Future payoff & Market object & Current loading & Required non-age information \\
\midrule
Risk-free saving & Next-period survival-contingent asset payoff & $1+r$ & Borrowing/Euler wedge & Transition-row type \\
Risky asset & State-contingent payoff vector & Asset price or state-price vector & Portfolio constraints & Payoff regime and transition row \\
Labor tax & Current labor or resource perturbation & After-tax wage & Labor wedge & Productivity, earnings, and tax regime \\
Liquid-illiquid saving & Liquid and illiquid portfolio payoff & Asset prices and adjustment cost & Portfolio and adjustment constraints & Illiquid state, adjustment regime, and transition row \\
\bottomrule
\end{tabular}%
}
\end{table}

\subsection{General verification for convex household problems}

Consider now a finite-life household with state $x=(a,\zeta,w)$, age, nonresource state, and scalar resources. At age $a$ it solves
\begin{equation}\label{eq:ha-class-value}
V_a(w,\zeta;m)=\max_{d\in D_a(\zeta,w)}\left\{u_a(c_a(d,w,\zeta)+m,\zeta)+\beta p_a\int V_{a+1}(W_a(d,w,\zeta,\zeta'),\zeta';0)P_a(\zeta,w,d,\dd\zeta')\right\}
\end{equation}
subject to convex constraints
\begin{equation}\label{eq:ha-class-constraints}
g_{aj}(d,w,\zeta)\le0,
\qquad j=1,\ldots,J_a .
\end{equation}
Let $\contv_a(w,\zeta)=\partial V_a(w,\zeta;m)/\partial m|_{m=0}>0$.

\begin{assumption}[Convex age-structured household class]\label{ass:ha-class}
For each $(a,\zeta)$, feasible sets are nonempty and convex, the objective in
\Cref{eq:ha-class-value} is concave, and the constraints in
\Cref{eq:ha-class-constraints} satisfy Slater's condition: relative to the affine
constraints, there is a feasible choice that satisfies the inequality constraints
strictly. Flow utility is differentiable and strictly concave, $\contv_a>0$, and
$w\mapsto \contv_a(w,\zeta)$ is continuous and strictly monotone on the support
region used for matching. For margin $M$, fix a verified measurable tuple $\Psi_a^M(\zeta)$ such that equality of
$(\Psi_a^M,\contv_a)$, after age alignment, identifies the represented future leg, the
gross market compensation object, and the relevant loadings.
\end{assumption}

\begin{remark}[Constructive certificate for richer state spaces]\label{rem:constructive-certificate}
A primitive certificate lists the represented leg, market price, and KKT loadings, then lets $\Psi_a^M$ be the minimal non-age tuple that identifies them once $\contv_a$ and age alignment are fixed.
\end{remark}

\begin{definition}[Verification protocol for the current map]\label{prop:model-derived-map}
For margin $M$, define
\begin{equation}\label{eq:model-derived-map}
\Comp_a^M(w,\zeta):=\Psi_a^M(\zeta),\qquad
\Psi_a^M(\zeta)=\min\{\text{non-age primitives identifying }(L^M,\Qterm^M,\Xi^M,\operatorname{supp} L^M)\},
\end{equation}
where the minimum is in the partition-refinement order on the maintained support.
On maintained support, a cross-age current match requires
\[
\Comp_a^M(w,\zeta)=\Comp_{a'}^M(\widetilde w,\widetilde\zeta),
\qquad
\contv_a(w,\zeta)=\contv_{a'}(\widetilde w,\widetilde\zeta).
\]
Every strict coarsening requires a separate sufficiency check, whereas any refinement remains sufficient and can only reduce the admissible match set.
\end{definition}
\begin{remark}[Represented margin dependence and scope]\label{rem:margin-dependence}
The one-step risk-free saving margin is a benchmark. Risky assets, labor supply, promised utility, or longer legs may enlarge the tuple and require a separate verification.
\end{remark}

\subsection{Longer comparison chains}

For a chain of exact fiber comparison tuples $\{\pi_\ell\}_{\ell=0}^{n-1}$, define
\[
\mathcal D^S(\pi_{0:n-1}):=\prod_{\ell=0}^{n-1}\mathcal D^S(\pi_\ell).
\]

\begin{corollary}[Blockwise multiplication on a fixed chain]\label{cor:chain-decomposition}
For any exact fiber comparison tuples $\{\pi_\ell\}_{\ell=0}^{n-1}$, singleton or finite-portfolio future legs,
\begin{equation}\label{eq:chain-decomposition}
\mathcal D^S(\pi_{0:n-1})=\prod_{\ell=0}^{n-1}\Mterm(\pi_\ell)\Rterm^S(\pi_\ell).
\end{equation}
If every future leg is a singleton and, for each $\ell<n-1$, the dated future older cell in block $\ell$ is the dated current older cell in block $\ell+1$, then with $(t_n,x^{\mathrm o}_{t_n,n})=(s_{n-1},x^{\mathrm o}_{s_{n-1},n-1})$,
\begin{equation}\label{eq:chain-decomposition-telescoped}
\mathcal D^S(\pi_{0:n-1})
=
\frac{\smw_{t_0}(x^{\mathrm o}_{t_0,0})}{\smw_{t_n}(x^{\mathrm o}_{t_n,n})}
=
\prod_{\ell=0}^{n-1}\Mterm(\pi_\ell)\Rterm^S(\pi_\ell).
\end{equation}
The product form also holds for portfolio-valued legs, with bounds multiplying across links and log errors adding.
\end{corollary}

\begin{remark}[Chain identity versus path consistency]\label{rem:chain-not-path}
\Cref{cor:chain-decomposition} is a product identity along a specified chain. 
\Cref{thm:path-independence} imposes the stronger requirement that $\log\ndw$ be consistent across all paths in the comparison network. 
Once exact equality in $\contv_t$ is replaced by approximate matching, finite-grid recursions depend on the chosen network.
\end{remark}

\begin{remark}[Relation to Eden's benchmark]
Dropping $\contv_t$ gives same-$\Comp_t$ matching. Dropping $\Comp_t$ as well, or taking it to be trivial, gives pure age matching. Exact fibers coincide with age matching only for cells that share the verified code and private envelope.
\end{remark}

\subsection{Comparison-system measure}

Now set $\Psi_t(x):=(\Comp_t(x),\contv_t(x))$, $\nu_t:=\Psi_{t\#}\dist_t$, and let $\{\dist_{t,g}\}_g$ be the regular conditional law. An \emph{exact fiber comparison system} is a measurable subprobability kernel $K_t(g,\dd\pi)$ over admissible bridge-verified exact fiber tuples, whose normalized current marginals are absolutely continuous with respect to $\dist_{t,g}$ on fibers with $K_t(g,\Pi)>0$ (unsupported fibers receive zero mass). It induces
\[
\Theta_t(\dd\pi):=\int K_t(g,\dd\pi)\,\nu_t(\dd g).
\]
Note that almost-everywhere and positive-measure statements below are with respect to $\Theta_t$. This prevents the comparison system from concentrating on representatives outside the full measure versions used for conditional residual neutrality.

The economic implication of this comparison system is stated in \Cref{cor:trilemma}.

\subsection{Approximate networks and convergence}
\begin{proposition}[Approximate chain identity and endpoint mismatch]\label{prop:approx-chain-bound}
Fix date $t$ and a directed current cell chain $c=(x^0,x^1,\ldots,x^n)$. For each edge $\ell$, suppose there is a bridge-verified admissible approximate tuple
\[
\bar\pi_\ell=(t,x^\ell,x^{\ell-1},x_{s_\ell}^{\mathrm o})
\]
with maintained current map equality, witness $x_{s_\ell}^{\mathrm o}\in\Gamma_{t,s_\ell}(x^{\ell-1})$, and mismatch
\[
\varepsilon_\ell:=|\log\contv_t(x^\ell)-\log\contv_t(x^{\ell-1})|.
\]
Define
\begin{equation}\label{eq:approx-edge-weight}
\widetilde g_t^\ell:=
\log\Qterm(\bar\pi_\ell)-\log\mathcal D^S(\bar\pi_\ell)+\log\CWedgeTerm(\bar\pi_\ell)
\end{equation}
and
\begin{equation}\label{eq:approx-chain-total}
\widetilde G_t(c):=\sum_{\ell=1}^n\widetilde g_t^\ell.
\end{equation}
Under the assumptions of \Cref{thm:impossibility},
\begin{equation}\label{eq:approx-chain-telescopes}
\widetilde G_t(c)=\log\smw_t(x^0)-\log\smw_t(x^n).
\end{equation}
Hence any second approximate chain $\widetilde c$ with the same endpoints satisfies
\begin{equation}\label{eq:approx-chain-disagreement}
\widetilde G_t(c)=\widetilde G_t(\widetilde c).
\end{equation}
Moreover,
\begin{equation}\label{eq:approx-chain-error}
\widetilde G_t(c)-\bigl(\log\ndw_t(x^0)-\log\ndw_t(x^n)\bigr)
=
\log\contv_t(x^0)-\log\contv_t(x^n),
\end{equation}
so
\begin{equation}\label{eq:approx-chain-error-bound}
\left|
\widetilde G_t(c)-\bigl(\log\ndw_t(x^0)-\log\ndw_t(x^n)\bigr)
\right|
\le
\sum_{\ell=1}^n\varepsilon_\ell.
\end{equation}
If each edge has mismatch at most $\bar\varepsilon$, the endpoint mismatch is at most $n\bar\varepsilon$.
\end{proposition}

\begin{proposition}[Finite-grid approximate-graph convergence]\label{prop:grid-convergence}
Let $G^0$ be an exact-fiber graph on a compact supported domain with finitely many connected components after the stated component anchors. Let $G_n$ be finite approximate graphs generated from partitions with mesh $h_n\downarrow0$ and matching tolerances $\varepsilon_n\downarrow0$. Suppose: (i) every exact edge and target cell has approximating finite-grid counterparts; (ii) no sequence of approximate edges converges outside the exact edge relation; (iii) the composite edge label $\log\Qterm-\log\bar{\mathcal D}^S+\log\CWedgeTerm$ and the target weights are continuous on the supported domain; (iv) the normalized feasible sets for the finite programs are uniformly bounded after component anchoring; and (v) the anchored finite feasible sets are nonempty for all sufficiently large $n$ and converge inner semicontinuously to the exact anchored set. Embed each finite potential as the partition-cell piecewise-constant function on the exact supported domain, with anchored components identified in $L^\infty$ modulo null sets. Then, the feasible potential sets converge in Hausdorff distance in this common metric space to the exact anchored feasible set, and the identified target interval in \Cref{eq:finite-graph-objective} satisfies
\[
\underline\Delta_{AB}^{(n)}\to \underline\Delta_{AB}^{0},
\qquad
\overline\Delta_{AB}^{(n)}\to \overline\Delta_{AB}^{0}.
\]
Thus, shrinking the resource grid and matching tolerance recovers the exact-fiber interval when support and component structure are stable. In KKT environments the regime code is explicit: borrowing-kink cells are either retained as boundary cells or removed in a separate support check, not treated as smooth interiors.
\end{proposition}

\begin{remark}[From exact potentials to finite grid recursions]\label{rem:exact-vs-approx-graph}
The exact theorem gives a potential for $\log\ndw$. Approximate chain totals telescope in raw welfare weights, and the endpoint mismatch in the private scale controls their distance from that residual potential.
\end{remark}

\begin{remark}[From tuplewise to class level restrictions]\label{rem:tuplewise-regularity}
Pairwise $\varepsilon$ bounds become class-level bounds with regularity of a declared class summary of $\ndw_t$ across nearby private envelopes. Here $\bar\eta_t(u,e^z)$ denotes a maintained positive measurable selection, average, or envelope over $\{x:\Comp_t(x)=u,\ \contv_t(x)=e^z\}$; if $z\mapsto\log\bar\eta_t(u,e^z)$ is $L$-Lipschitz, moving $\varepsilon$ in $z=\log\contv_t$ adds at most $L\varepsilon$. For multidimensional states, a declared metric $d$ on sufficient codes yields a relaxed edge when $d(x,x')\le\delta$, with graph band $\varepsilon_e+L\delta$; exact matching is $\delta=0$.
\end{remark}


\section{Additional Quantitative Results}\label{app:additional-quantitative}

The tables below report the full comparison set and the unit, support, wedge, grid, and richer-state exercises summarized in the main text.

\begin{table}[htbp]
\centering
\caption{Full identification results for the focal target comparison}
\label{tab:full-main-diagnostic-audit}
\scriptsize
\renewcommand{\arraystretch}{1.05}
\resizebox{0.99\textwidth}{!}{%
\begin{tabular}{>{\raggedright\arraybackslash}p{0.24\textwidth}>{\raggedright\arraybackslash}p{0.13\textwidth}>{\raggedright\arraybackslash}p{0.14\textwidth}>{\raggedright\arraybackslash}p{0.10\textwidth}>{\raggedright\arraybackslash}p{0.09\textwidth}>{\raggedright\arraybackslash}p{0.13\textwidth}>{\raggedright\arraybackslash}p{0.17\textwidth}}
\toprule
Maintained system & Support & Unit & Normalization & Wedge & Identified sign? & Interval \\
\midrule
Age-only & full & implicit dollars & none & none & No & [-0.04044,\ 0.15093] \\
Age-only & common & implicit dollars & none & none & No & [-0.00967,\ 0.07886] \\
Matched states & common & current goods & fixed & KKT & Yes, conditional & [-0.00819,\ -0.00298] \\
Matched states & common & current goods & free & KKT & No & [$-\infty$,\ $+\infty$] \\
Matched states & common & current goods & outer & KKT & No & [-2.16877,\ 2.15760] \\
Largest common component & common comp. & current goods & within comp. & KKT & Yes, within comp. & [0.00004,\ 0.00082] \\
Unit-rescaled network & rebuilt & age-state & fixed & KKT & No & [-0.48701,\ 0.50273] \\
Richer-state system & rebuilt & current goods & fixed & KKT & No & [-0.03237,\ 0.06770] \\
\bottomrule
\end{tabular}%
}
\begin{minipage}{0.99\textwidth}
\footnotesize Notes. Intervals are supported contrasts in normalized welfare weights; positive $\Delta_{AB}$ favors target $A$. An identified sign is conditional on the maintained row. Unrestricted component normalizations are unbounded. The outer-envelope row restricts component gaps by the log range of private marginal values on maintained support.
\end{minipage}
\end{table}


\subsection{Omitted support, unit rebuilds, and KKT-consistent wedges}\label{subsec:omitted-support-normalization-wedge}

Table~\ref{tab:unit-wedge-protocol} treats the private unit, continuation adjustment, and constraint wedge as identification primitives. The supported comparison leaves omitted-support, private-unit, continuation-adjustment, and wedge-implementation layers; support shares are over the focal target domain. The current-goods unit is tied to the planner's physical resource constraint, while the age-state rescaled unit rebuilds support, candidates, components, and intervals through $\widetilde\contv^{(\theta)}=\contv/\bar\contv_a(k)^\theta$. The baseline sign is certified for $\theta\le0.50$ and continuation-adjustment range $|\Delta\log\mathfrak B|\le0.001$, but not by $\theta=0.75$ or $\theta=1$.

My raw-Euler row is a numerical stress test. It shows what happens if finite-grid Euler residuals are propagated as economic wedges. The maintained object uses the KKT-consistent projection: set $\widehat\lambda_a=\max\{u'(c_a)-E_a,0\}$ on boundary cells and $0$ on strict interiors up to tolerance, and use $\Lambda=[1-\widehat\lambda/u'(c)]^{-1}$ clipped at one before graph recursion. On the baseline margin, strict interior cells carry $\BaselineStrictInteriorShare$ of positive mass and boundary-loading cells $\BaselineBoundaryLoadingShare$. The implemented wedge mean and 95th percentile are $\BaselineImplementedWedgeMean$ and $\BaselineImplementedWedgePNinetyFive$.

\begin{table}[htbp]
\centering
\caption{Unit path, omitted support, and wedge diagnostics}
\label{tab:unit-wedge-protocol}
\footnotesize
\renewcommand{\arraystretch}{1.04}
\resizebox{0.99\textwidth}{!}{%
\begin{tabular}{>{\raggedright\arraybackslash}p{0.28\textwidth}c>{\raggedright\arraybackslash}p{0.16\textwidth}ccccc}
\toprule
Specification & $\theta$ & Network rule & Target support share & Components & Point & Interval or bound & Identified sign \\
\midrule
Current-goods unit & 0.00 & tolerance graph & 0.7472 & 3737 & -0.00481 & [-0.00819,\ -0.00298] & negative \\
Bounded continuation adjustment, $\bar b=0.001$ & none & fixed-network range & 0.7472 & 3737 & -0.00481 & [-0.00919,\ -0.00198] & negative \\
Unit path, $\theta=0.25$ & 0.25 & rebuilt tolerance network & 0.7573 & 62 & -0.00457 & [-0.01611,\ -0.00281] & negative \\
Unit path, $\theta=0.50$ & 0.50 & rebuilt tolerance network & 0.7888 & 62 & -0.02694 & [-0.08342,\ -0.02136] & negative \\
Unit path, $\theta=0.75$ & 0.75 & rebuilt tolerance network & 0.8166 & 47 & -0.03167 & [-0.28857,\ 0.30268] & no \\
Age-state rescaled unit & 1.00 & rebuilt tolerance network & 0.9895 & 37 & 0.04076 & [-0.48701,\ 0.50273] & no \\
Direct raw Euler proxy wedge & none & diagnostic five-nearest network & 0.7472 & 77 & 0.03032 & [-0.97832,\ 0.98022] & no \\
Illustrative lower-tail completion & none & completion bound & 1.0000 & not applicable & not applicable & [-0.02902,\ -0.02402] & maintained completion \\
\bottomrule
\end{tabular}%
}
\begin{minipage}{0.99\textwidth}
\footnotesize Notes. Target support share is over focal target mass; candidate share in \Cref{tab:grid-candidate-sensitivity} is over tolerance candidates. The continuation-adjustment row reports the induced target-average range $|\Delta\log\BridgeCorr|\le0.001$ on the fixed network; the fixed-normalization sign is lost only when this range reaches $0.00298$. Unit-path rows rebuild the network. The raw Euler proxy uses a diagnostic five-nearest network and unprojected finite-grid Euler residuals.
\end{minipage}
\end{table}


\subsection{Locality, grid discipline, and richer states}\label{subsec:locality-grid-states}

Table~\ref{tab:grid-candidate-sensitivity} reports grid refinement and tolerance-first matching. Its candidate support share is computed over the tolerance-candidate network, and not over the focal target domain in \Cref{tab:unit-wedge-protocol}; note also that each tolerance row is a maintained finite-network estimand, not by itself a joint proof of exact-fiber convergence. Grid mismatch is small. The maintained tolerance row uses $\varepsilon=0.005$ and an exhaustive search cap $K=58$. Holding targets fixed, higher curvature $\sigma=2$ gives $-0.00444$ with $[-0.00607,-0.00348]$ and $\beta=0.97$ gives $-0.00137$ with $[-0.00180,-0.00110]$. Lower patience $\beta=0.94$ gives $0.13629$ with $[0.11113,0.14976]$ and lower $r=0.02$ gives $0.07253$ with $[0.05591,0.08221]$. Holding the calibration fixed and shifting target windows gives $[0.00090,0.01444]$ for ages $20$ to $29$ versus $40$ to $49$, the focal $[-0.00819,-0.00298]$, and $[-0.00929,-0.00516]$ for ages $30$ to $39$ versus $50$ to $59$.

\begin{table}[htbp]
\centering
\caption{Grid refinement and tolerance-first matching diagnostics}
\label{tab:grid-candidate-sensitivity}
\footnotesize
\renewcommand{\arraystretch}{1.05}
\resizebox{0.99\textwidth}{!}{%
\begin{tabular}{lcccccc}
\toprule
\multicolumn{7}{l}{\textit{Panel A. Asset-grid refinement}} \\
Asset points & Mean $|\Delta \log\upsilon|$ & 95th percentile $|\Delta \log\upsilon|$ & Mean wedge & 95th percentile wedge & $\widehat{\Spread}$ & $\widehat{\DemDist}$ \\
\midrule
250 & 0.0006 & 0.0033 & 1.0120 & 1.0842 & 3.036 & 5.500 \\
800 & 0.0004 & 0.0017 & 1.0119 & 1.0839 & 3.027 & 5.231 \\
1200 & 0.0003 & 0.0013 & 1.0118 & 1.0834 & 3.023 & 5.195 \\
\midrule
\multicolumn{7}{l}{\textit{Panel B. Tolerance-first current-state network}} \\
Tolerance $\varepsilon$ & Exhaustive cap $K$ & Exhaustive? & Mean mismatch & 95th percentile mismatch & Candidate support share & Interval \\
\midrule
0.001 & 24 & yes & 0.00042 & 0.00094 & 0.5686 & $[-0.00374,\ -0.00350]$ \\
0.002 & 34 & yes & 0.00088 & 0.00188 & 0.7438 & $[-0.00664,\ -0.00506]$ \\
0.005 & 58 & yes & 0.00231 & 0.00471 & 0.8993 & $[-0.00819,\ -0.00298]$ \\
0.010 & 82 & yes & 0.00475 & 0.00942 & 0.9268 & $[-0.01055,\ -0.00219]$ \\
\bottomrule
\end{tabular}%
}
\begin{minipage}{0.99\textwidth}
\footnotesize Notes. Panel B fixes the tolerance band. The candidate support share is computed on the network admitted by that band, not on the target domain alone. $\widehat{\Spread}$ is the finite log spread of the recovered multipliers within a fiber; $\widehat{\DemDist}$ is the corresponding transport distance across ages. The cap $K$ is the first value at which the search has reached all retained pairs. The maintained row sets $\varepsilon=0.005$.
\end{minipage}
\end{table}


For the maintained row, the result is not a near miss. At $\varepsilon=0.005$, the normalized finite program is feasible, and the parallel, cycle, band, and relaxation diagnostics are all $0$ at the reported precision. The borrowing kink check then removes target cells whose common support depends on the lower bound. Retained focal support falls from $0.7472$ to $0.7177$, and the identified interval becomes $[-0.00159,0.00090]$. The estimand supported by the sample has changed; the finite grid has not failed. The exercise isolates dependence on support at the boundary from numerical infeasibility or noise in the raw Euler equations.

The exercise with liquid and illiquid assets yields overlap 
$\LiquidIlliquidExactOverlap$ on the fine $k\times\varepsilon$ grid and 
$\LiquidIlliquidTercileOverlap$ when assets are grouped into terciles and 
earnings are matched. The corresponding mean and worst contamination are 
$\LiquidIlliquidTercileMeanKappaPlp$ and 
$\LiquidIlliquidTercileMaxKappaPlp$ percent log points. The point summaries are 
$\Delta_{AB}=\LiquidIlliquidTercilePointDelta$ with asset terciles and earnings, 
$\LiquidIlliquidEarningsOnlyPointDelta$ with earnings alone, and 
$\LiquidIlliquidAgeOnlyPointDelta$ with age alone. The corresponding network 
intervals are 
$[\LiquidIlliquidTercileGraphLower,\LiquidIlliquidTercileGraphUpper]$, 
$[\LiquidIlliquidEarningsOnlyGraphLower,\LiquidIlliquidEarningsOnlyGraphUpper]$, 
and 
$[\LiquidIlliquidAgeOnlyGraphLower,\LiquidIlliquidAgeOnlyGraphUpper]$ percent 
log points. Since all three intervals contain $0$, the richer state space does not identify a stronger sign. The fallback is the rule described above: a match 
at distance $\delta$ is allowed only after the bounds are widened by the declared 
interpolation allowance $L\delta$.




\section{Aggregation and Coarsening}\label{sec:geometry}

This appendix section converts the pairwise restriction into exact-fiber spread, mismatch-adjusted bounds, and contamination from coarsening. Regular conditional laws exist because $\Pop=\Age\times\State$ and the comparison-map codomain are maintained standard Borel.

\subsection{Spread as an aggregation of pairwise gaps}

Let
\[
\Psi_t(x):=(\Comp_t(x),\contv_t(x)),
\qquad
\nu_t:=\Psi_{t\#}\dist_t,
\]
and fix a regular conditional law $\{\dist_{t,g}\}_{g}$ for $X_t\sim\dist_t$ given $\Psi_t(X_t)=g$.

\begin{definition}[Fiberwise and global demographic spread]\label{def:spread}
For date $t$ and fiber value $g$, define
\begin{equation}\label{eq:spread-fiber}
\Spread_t(\ndw\mid g):=
\esssup_{\dist_{t,g}}\log\ndw_t-
\essinf_{\dist_{t,g}}\log\ndw_t.
\end{equation}
The global spread is
\begin{equation}\label{eq:spread}
\Spread_t(\ndw):=
\esssup_{g\sim\nu_t}\Spread_t(\ndw\mid g).
\end{equation}
On finite grids, these are ordinary within fiber ranges and their maximum. Thus $\Spread_t(\ndw)=0$ iff full-fiber conditional residual age neutrality holds at date $t$; tuple-relative neutrality gives the lower-bound statements below without requiring unused cells in the fiber to be neutral.
\end{definition}

\begin{definition}[Exact fiber admissible comparison system]\label{def:comparison-system}
An \emph{exact fiber admissible comparison system} anchored at $t$ is a measurable subprobability kernel $K_t(g,\dd\pi)$ from fiber values into bridge-verified comparison tuples such that, for $\nu_t$ almost everywhere $g$, either $K_t(g,\Pi)=0$ or the normalized kernel is supported on tuples whose two current cells lie in fiber $g$ and whose induced current old and current young marginals are absolutely continuous with respect to $\dist_{t,g}$. Unsupported fibers are assigned the zero kernel.
\end{definition}

For such a system, set $\underline g_t(g;K_t)=0$ when $K_t(g,\Pi)=0$; otherwise define
\begin{equation}\label{eq:g-lower-profile}
\underline g_t(g;K_t):=
\esssup_{\pi\sim K_t(g,\cdot)/K_t(g,\Pi)}
[\log\Qterm(\pi)-\log\mathcal D^S(\pi)]_+
\end{equation}
This wedge-free profile is conservative. With a verified KKT wedge and bridge correction, replace the bracket by $[\log\Qterm-\log\mathcal D^S+\log\CWedgeTerm]_+$. The numerics use the conservative object unless a wedge implementation is named. Define
\begin{equation}\label{eq:g-lower}
\underline g_t(K_t):=
\esssup_{g\sim\nu_t}\underline g_t(g;K_t).
\end{equation}

\begin{proposition}[Fiberwise lower bound profile]\label{prop:fiberwise-bound}
Fix $t$ and an exact fiber admissible comparison system $K_t$. Under \Cref{ass:benchmark,ass:localized-cell,ass:local-derivatives,ass:supporting-prices,ass:represented-domain-consistency,ass:represented-family-bridge} and nonpaternalism on $M$,
\begin{equation}\label{eq:fiberwise-bound}
\Spread_t(\ndw\mid g)\ge \underline g_t(g;K_t)
\qquad\text{for }\nu_t\text{ almost everywhere }g.
\end{equation}
\end{proposition}

\begin{corollary}[Aggregated spread bound]\label{cor:aggregate-spread}
Under the same assumptions,
\begin{equation}\label{eq:aggregate-spread}
\Spread_t(\ndw)\ge \underline g_t(K_t).
\end{equation}
\end{corollary}

Finite grids add approximation because exact equality of $\contv_t$ is typically unavailable. \Cref{prop:grid-convergence} gives the corresponding grid-refinement discipline when support and component structure are stable.

\subsection{Approximate matching and contamination under coarsened maps}

Let $\Comp_t^\ast$ be the verified current-state map and $\bar\Comp_t=\phi_t\circ\Comp_t^\ast$ a coarsening. And let $B_t:\RR_+\to\mathcal B_t$ bin the private envelope scale (the identity in exact applications, weighted $\log\contv$ bins in the numerics). Define
\[
\bar\Psi_t(x):=(\bar\Comp_t(x),B_t(\contv_t(x))),
\qquad
\bar\nu_t:=\bar\Psi_{t\#}\dist_t,
\]
and fix conditional laws $\bar\dist_{t,\bar g}$. The coarse fiber spread is
\[
\Spread_t^{\bar\Comp}(\ndw\mid\bar g):=
\esssup_{\bar\dist_{t,\bar g}}\log\ndw_t-
\essinf_{\bar\dist_{t,\bar g}}\log\ndw_t.
\]

\begin{definition}[Operative coarse comparison system]
An \emph{operative coarse comparison system} anchored at $t$ and $\bar\Psi_t$ is a measurable subprobability kernel $\bar K_t(\bar g,\dd\pi)$ into bridge-verified ordered tuples retaining the age ordering and future-leg condition of \Cref{def:comparison-triple}, with both current cells in coarse fiber $\bar g$ and normalized current marginals absolutely continuous with respect to $\bar\dist_{t,\bar g}$ when $\bar K_t(\bar g,\Pi)>0$. Unsupported coarse fibers receive zero mass. The term \emph{comparison tuple} remains reserved for exact fiber objects.
\end{definition}

All $\bar K_t$-essential suprema below normalize the relevant positive-mass section and equal $0$ when it has zero mass.

For tuples in $\bar K_t$, define
\begin{equation}\label{eq:coarse-H-term}
\mathcal H(\pi):=
\log\Qterm(\pi)-\log\mathcal D^S(\pi)+\log\CWedgeTerm(\pi),
\end{equation}
ages $a^{\mathrm y}(\pi),a^{\mathrm o}(\pi)$, and mismatch
\begin{equation}\label{eq:coarse-epsilon-term}
\varepsilon(\pi):=
|\log\contv_t(x_t^{\mathrm o})-\log\contv_t(x_t^{\mathrm y})|.
\end{equation}
For every such tuple the coarse-domain pairwise algebra gives
\begin{equation}\label{eq:coarse-domain-pairwise}
\log\frac{\ndw_t(x_t^{\mathrm y})}{\ndw_t(x_t^{\mathrm o})}
\ge \mathcal H(\pi)-\varepsilon(\pi),
\end{equation}
by the raw identity in \Cref{prop:coarse-matching} and the represented bridge $\Rterm^S=\Qterm\CWedgeTerm$. The age-pair-conditioned mismatch-adjusted profile is
\begin{equation}\label{eq:g-lower-profile-eps}
\underline g_t^\varepsilon(a^{\mathrm y},a^{\mathrm o},\bar g;\bar K_t)
:=
\esssup_{\pi\sim\bar K_t(\bar g,\cdot):\,a^{\mathrm y}(\pi)=a^{\mathrm y},\,a^{\mathrm o}(\pi)=a^{\mathrm o}}
[\mathcal H(\pi)-\varepsilon(\pi)]_+,
\end{equation}
with value $0$ when the age pair is absent.

\begin{definition}[Approximate sufficiency of a coarsened comparison map]\label{def:approx-sufficiency}
A coarsening $\bar\Comp_t$ is \emph{$\kappa$-approximately sufficient} relative to $\bar K_t$ if there is a measurable age-pair-conditioned benchmark $\bar h_t$ and $\kappa\ge0$ such that
\begin{equation}\label{eq:approx-sufficiency-def}
|\mathcal H(\pi)-\bar h_t(a^{\mathrm y}(\pi),a^{\mathrm o}(\pi),\bar g)|\le\kappa
\end{equation}
for $\bar\nu_t$ almost everywhere $\bar g$ and $\bar K_t(\bar g,\cdot)$ almost everywhere $\pi$.
\end{definition}

Define
\begin{equation}\label{eq:coarse-eps-radius}
\bar\varepsilon_t(a^{\mathrm y},a^{\mathrm o},\bar g):=
\esssup_{\pi\sim\bar K_t(\bar g,\cdot):\,a^{\mathrm y}(\pi)=a^{\mathrm y},\,a^{\mathrm o}(\pi)=a^{\mathrm o}}\varepsilon(\pi),
\end{equation}
again $0$ if the age pair is absent. Let $\widetilde\Psi_t(\pi):=(a^{\mathrm y}(\pi),a^{\mathrm o}(\pi),\bar\Psi_t(x_t^{\mathrm y}))$ and let $\widetilde\nu_t^{\bar K}$ be the finite pushforward measure induced by $\bar\nu_t(\dd\bar g)\bar K_t(\bar g,\dd\pi)$; if its total mass is $0$, the following statements are vacuous.

\begin{proposition}[Coarsening contamination bound]\label{prop:approx-sufficiency-bound}
If $\bar\Comp_t$ is $\kappa$-approximately sufficient relative to $\bar K_t$, then under the maintained global assumptions and nonpaternalism on $M$,
\begin{equation}\label{eq:coarse-spread-bound}
\Spread_t^{\bar\Comp}(\ndw\mid\bar g)
\ge
[\bar h_t(a^{\mathrm y},a^{\mathrm o},\bar g)-\bar\varepsilon_t(a^{\mathrm y},a^{\mathrm o},\bar g)-\kappa]_+
\end{equation}
for $\widetilde\nu_t^{\bar K}$ almost everywhere $(a^{\mathrm y},a^{\mathrm o},\bar g)$. Consequently,
\begin{equation}\label{eq:coarse-spread-global-bound}
\esssup_{(a^{\mathrm y},a^{\mathrm o},\bar g)\sim\widetilde\nu_t^{\bar K}}\Spread_t^{\bar\Comp}(\ndw\mid\bar g)
\ge
\esssup_{(a^{\mathrm y},a^{\mathrm o},\bar g)\sim\widetilde\nu_t^{\bar K}}
[\bar h_t(a^{\mathrm y},a^{\mathrm o},\bar g)-\bar\varepsilon_t(a^{\mathrm y},a^{\mathrm o},\bar g)-\kappa]_+.
\end{equation}
Relative to \Cref{cor:approx-impossibility}, $\kappa$ is the extra loss from coarsening the current map.
\end{proposition}

\begin{remark}[Operational sign certification]\label{rem:sign-certification}
A comparison system is sign-informative only when \Cref{eq:coarse-spread-bound} is strictly positive on operative fibers and the finite graph has a nonempty feasible potential set. Report overlap, mismatch, contamination, edge/component counts, feasibility or cycle/parallel-label residuals, and interval width; if the interval contains $0$, refine the grid or partition, or declare non-certification.
\end{remark}

\begin{corollary}[Primitive Lipschitz control of contamination]\label{cor:primitive-kappa}
Suppose there is a metric $d_t$ on omitted current state objects and $L_t<\infty$ such that tuples $\pi,\pi'$ in the same coarse fiber and age pair satisfy
\[
|\mathcal H(\pi)-\mathcal H(\pi')|\le L_t d_t(\pi,\pi').
\]
This is primitive in omitted current-state traits only when the candidate social discount schedule and KKT wedge have separate Lipschitz control; otherwise it is a regularity assumption on composite $\mathcal H$.

Let
\[
\begin{aligned}
\Delta_t(a^{\mathrm y},a^{\mathrm o},\bar g):=
\sup\{d_t(\pi,\pi'):
&\ \pi,\pi'\text{ lie in the support of }\bar K_t(\bar g,\cdot),\\
&\ a^{\mathrm y}(\pi)=a^{\mathrm y}(\pi')=a^{\mathrm y},\
\ a^{\mathrm o}(\pi)=a^{\mathrm o}(\pi')=a^{\mathrm o}\},
\end{aligned}
\]
with value $0$ when empty. Then $\bar\Comp_t$ is $\kappa$-approximately sufficient for any
\[
\kappa\ge
\esssup_{(a^{\mathrm y},a^{\mathrm o},\bar g)\sim\widetilde\nu_t^{\bar K}}
\frac12 L_t\Delta_t(a^{\mathrm y},a^{\mathrm o},\bar g).
\]
\end{corollary}

Section~\ref{sec:illustrations} reports support, mismatch, contamination, diagnostic envelopes, and maintained point-summary diagnostics in that order; the replication outputs retain the fuller partition diagnostic.



\section{Proofs}\label{app:proofs}

This appendix collects algebra used outside the main text. Maintained primitives are \Cref{ass:benchmark,ass:localized-cell,ass:local-derivatives,ass:supporting-prices,ass:represented-domain-consistency,ass:represented-family-bridge}; result-specific hypotheses are added when invoked. Conditional objects are on positive support.

\subsection{Numerical accuracy and replication diagnostics}\label{app:numerical-accuracy}

The replication archive records the environment, run records, and generated table sources. The manuscript reports the diagnostics needed for the local quantitative test: grid refinement, private-unit mismatch, current-map dispersion, boundary classification, KKT wedge moments, bridge thresholds, raw-Euler stress rows, graph feasibility, bands, cycles, components, and minimal relaxation. The raw-Euler row in \Cref{tab:unit-wedge-protocol} is a stress test; the maintained wedge is the KKT-consistent projection.

\subsection{Regularity}

The proofs use measurability of comparison maps, continuation correspondences, and state-price kernels; existence of the stated derivatives; and regular conditional probabilities for \Cref{sec:geometry}. Standard Borel structure supplies the last step.

\begin{lemma}[Standard Borel disintegration]\label{lem:standard-borel-disintegration}
Let $\Psi_t(x)=(\Comp_t(x),\contv_t(x))$ be measurable from a standard Borel space $\Pop$ to a standard Borel codomain $\CompSpace\times\RR_+$. Then there is a regular conditional law $\{\dist_{t,g}\}_g$ for $X_t\sim\dist_t$ given $\Psi_t(X_t)=g$, unique for $\nu_t:=\Psi_{t\#}\dist_t$ almost every $g$, and
\[
\int f(x)\dist_t(\dd x)=\int\left[\int f(x)\dist_{t,g}(\dd x)\right]\nu_t(\dd g)
\]
for every integrable measurable $f$.
\end{lemma}

\begin{proof}
This is the standard regular conditional probability theorem for standard Borel spaces, followed by the law of iterated expectations.
\end{proof}

\subsection{Verification of current maps and bridges}

\begin{proof}[Proof of \Cref{prop:margin-sufficient-map}]
Part (i) is margin sufficiency. For (ii), merging distinct verified codes destroys inherited sufficiency; the coarsening either fails on merged support points or requires a new check. Part (iii) follows because equality of $(\Comp_t^M,W_t)$ implies equality of $\Comp_t^M$, and adding $W_t$ only removes matches.
\end{proof}

\begin{proof}[Proof of \Cref{cor:one-asset-map}]
Fix an age and transition-row class, and let $w_a(b,\varepsilon_k)=y_a(\varepsilon_k)+(1+r)b$. The continuation problem and one-step risk-free leg are the same scalar problem in $w_a$. Strict monotonicity of $\contv_a$ in $w_a$ makes equal $\contv_a$ recover equal $w_a$, hence the same next asset and loading. Different transition rows use different Euler weights, so coarsening across rows requires a separate check; distinct rows give the earnings state.
\end{proof}

\begin{proof}[Proof of \Cref{cor:liquid-illiquid-bridge}]
Apply \Cref{prop:supporting-prices-broad} to the convex two-asset problem. The tuple augments liquid state with illiquid holdings and adjustment-cost derivatives; liquid and illiquid KKT multipliers enter through loadings. Convexity gives the supporting-price inequality, slackness gives the wedge-free private case, and \Cref{ass:represented-domain-consistency,ass:represented-family-bridge} plus nonpaternalism give the social equality as in \Cref{prop:implementability}.
\end{proof}

\begin{proof}[Proof of \Cref{lem:fiber-nonempty}]
Continuity and strict monotonicity make each $b\mapsto\contv_t(a,b,u)$ one-to-one onto an interval. Any value in the overlap has one preimage at each age, producing a shared matched set $(u,v)$. Admissibility of a local comparison additionally requires the continuation-support condition imposed in \Cref{def:comparison-triple}.
\end{proof}

\begin{proof}[Proof of \Cref{prop:supporting-prices-broad}]
For a gross future leg $R$, the private derivative is
\[
-\contv_t(x_t^{\mathrm y})+\frac{R}{\Qterm(\pi)}\left(\contv_t(x_t^{\mathrm y})-\sum_j\lambda_j(\pi)\alpha_j(\pi)\right).
\]
Setting it to zero gives \Cref{eq:private-comp-broad}. Positivity of the bracket is the admissibility condition $\sum_j\lambda_j(\pi)\alpha_j(\pi)<\contv_t(x_t^{\mathrm y})$; $\Xi(\pi):=\sum_j\lambda_j(\pi)\alpha_j(\pi)$ is the reduced-form KKT loading.
\end{proof}

\begin{lemma}[Private compensation representation]\label{lem:supporting-prices}
For a represented family that removes one current unit and adds a priced future leg, the private indifference factor is $\widehat R^P(\pi)=\Qterm(\pi)\WedgeTerm(\pi)$ with $\WedgeTerm(\pi)\ge1$. Equality holds exactly when $\Xi(\pi)=0$; locally spanned slack perturbations have this property, and strict loading/complementarity within that class gives the converse.
\end{lemma}

\begin{proof}
This is \Cref{prop:supporting-prices-broad} with $\WedgeTerm(\pi)=[1-\Xi(\pi)/\contv_t(x_t^{\mathrm y})]^{-1}$.
\end{proof}

\begin{proof}[Proof of \Cref{prop:implementability}]
\Cref{lem:supporting-prices} gives $\widehat R^P=\Qterm\WedgeTerm$. Nonpaternalism and \Cref{ass:represented-family-bridge} identify this factor with $\Rterm^S$, so $\Rterm^S=\Qterm\WedgeTerm\ge\Qterm$. Equality requires $\Xi=0$; locally spanned slack perturbations are the primitive wedge-free case, with strict complementarity giving the converse.
\end{proof}

\begin{proof}[Proof of \Cref{cor:market-value}]
If the leg is locally purchasable at state-price $q_{t,s}$, then $\Qterm(\pi)=q_{t,s}(x_s^{\mathrm o}\mid x_t^{\mathrm y})^{-1}$. Substitute this into \Cref{eq:R-market-wedge}. The equality statement is the equality statement in \Cref{prop:implementability}.
\end{proof}

\begin{proof}[Proof of \Cref{prop:policy-margin-reduction}]
By \Cref{prop:pcu-representation}, each derivative equals its $\PVEU$ target integral minus the common financing integral. Subtracting cancels financing and gives \Cref{eq:policy-margin-reduction}. Equal target mass makes the sign the sign of target-average differences.
\end{proof}

\begin{proof}[Proof of \Cref{prop:common-support-selection}]
Decompose the signed target measure over $S$ and $S^c$. The $S$ term is identified; omitted mass can overturn signs. Bounds $\ndw_t\in[\underline\eta,\overline\eta]$ give the displayed sharp outer interval by assigning endpoints to omitted $A$ and $B$ mass.
\end{proof}

\subsection{Pairwise algebra and graph restrictions}

\begin{proof}[Proof of \Cref{prop:coarse-matching}]
Start from $\mathcal D^S=(\smw_t(x_t^{\mathrm o})/\smw_t(x_t^{\mathrm y}))\Rterm^S$ and substitute $\smw_t=\ndw_t\contv_t$ without canceling $\contv_t$. Taking logs gives \Cref{eq:app-coarse-identity}.
\end{proof}

\begin{proof}[Proof of \Cref{cor:approximate-matching}]
The mismatch assumption bounds the last term in \Cref{eq:app-coarse-identity} between $-\varepsilon$ and $\varepsilon$.
\end{proof}

\begin{proof}[Proof of \Cref{prop:coarse-ambiguity}]
Normalize $\smw_s(x_s^{\mathrm o})=\smw_t(x_t^{\mathrm o})=\contv_t(x_t^{\mathrm y})=1$, set $\smw_t(x_t^{\mathrm y})=e^\gamma$ and $\contv_t(x_t^{\mathrm o})=e^{-\delta}$, and take $\delta>\gamma>0$. Then \Cref{eq:app-coarse-identity} gives residual log ratio $\gamma-\delta<0$, while an exact fiber has private-scale ratio one and residual log ratio $\gamma>0$.
\end{proof}

\begin{proof}[Proof of \Cref{cor:chain-decomposition}]
Apply \Cref{thm:mr} block by block and multiply. Singleton future legs linked to the next older current cell cancel intermediate welfare weights and give the endpoint ratio; portfolio legs need not. Product bounds and log-error addition follow.
\end{proof}

\begin{proof}[Proof of \Cref{thm:impossibility}]
\Cref{thm:mr} gives $\ndw_t(x_t^{\mathrm y})/\ndw_t(x_t^{\mathrm o})=\Rterm^S(\pi)/\mathcal D^S(\pi)$. \Cref{prop:implementability} sets $\Rterm^S=\Qterm\WedgeTerm$ in the baseline, yielding the displayed log equality; with bridge correction replace $\WedgeTerm$ by $\CWedgeTerm$. Since $\WedgeTerm\ge1$, $\mathcal D^S<\Qterm$ implies positive residual tilt, contradicting neutrality on that fiber.
\end{proof}

\begin{proof}[Proof of \Cref{prop:exact-edge-weight}]
This is the displayed equality in the proof of \Cref{thm:impossibility}, written as a directed edge weight. With $\BridgeCorr\neq1$, the same equality uses the corrected bridge relation $\Rterm^S=\Qterm\WedgeTerm\BridgeCorr$. The weak lower bound that drops $\log\WedgeTerm\ge0$ is a baseline $\BridgeCorr\equiv1$ statement only.
\end{proof}

\begin{proof}[Proof of \Cref{prop:imposed-discount-rationalization}]
Necessity is telescoping on closed walks. Same-edge witness labels must agree because each equals $\log\widetilde\eta(x^{\mathrm y})-\log\widetilde\eta(x^{\mathrm o})$. Conversely, impose parallel-label consistency, anchor each component, and define $\log\widetilde\eta(x)$ as minus the signed $\bar g$ path sum from the anchor to $x$. Cycle consistency gives path independence; exponentiation yields positive multipliers, unique up to component constants. Countable/measurable graphs use the same construction under integrable path sums.
\end{proof}

\begin{proof}[Proof of \Cref{thm:path-independence}]
\Cref{prop:exact-edge-weight} writes each edge as the tail-head difference in $\log\ndw$. Path sums telescope; cycles sum to zero; paths with common endpoints agree. A spanning-tree recursion with one anchor recovers the component potential, and changing the anchor rescales all levels by one constant.
\end{proof}

\begin{proof}[Proof of \Cref{cor:component-normalization}]
Choose an anchor level in each component and assign every other level by summing exact edge weights along any path from the anchor. \Cref{thm:path-independence} makes the assignment path independent. A different anchor multiplies all component levels by the same constant, leaving ratios unchanged.
\end{proof}

\begin{proof}[Proof of \Cref{cor:approx-impossibility}]
Combine \Cref{cor:approximate-matching} with $\Rterm^S=\Qterm\CWedgeTerm$ to get $\log(\ndw_t(x_t^{\mathrm y})/\ndw_t(x_t^{\mathrm o}))\ge \log\Qterm-\log\mathcal D^S+\log\CWedgeTerm-\varepsilon$. The stated sufficient condition makes the right side positive; in the baseline, dropping $\log\WedgeTerm\ge0$ gives the simpler sufficient condition.
\end{proof}

\begin{proof}[Proof of \Cref{prop:approx-chain-bound}]
For each approximate edge, \Cref{ass:represented-family-bridge} and \Cref{prop:implementability} give $\widetilde g_t^\ell=\log\smw_t(x^{\ell-1})-\log\smw_t(x^\ell)$. Sums telescope, so endpoint-equivalent chains agree. Since $\smw_t=\ndw_t\contv_t$, the residual-multiplier error is the endpoint private-scale difference, bounded by the edge mismatches sum.
\end{proof}

\begin{proof}[Proof of \Cref{cor:trilemma}]
If nonpaternalism holds and $\mathcal D^S<\Qterm$ on positive measure, \Cref{thm:impossibility} forces positive residual age tilt, contradicting neutrality. Conversely, nonpaternalism plus neutrality sets the left side of \Cref{eq:impossibility-ratio} to zero, so $\mathcal D^S=\Qterm\WedgeTerm\ge\Qterm$ almost everywhere.
\end{proof}

\subsection{Aggregation, contamination, and computation}

\begin{proof}[Proof of \Cref{prop:fiberwise-bound}]
For every tuple in fiber $g$, \Cref{thm:impossibility} gives \Cref{eq:impossibility-ratio}. Absolute continuity of the current marginals with respect to $\dist_{t,g}$ bounds the right side by the within-fiber essential range of $\log\ndw_t$; taking the tuple essential supremum gives the result.
\end{proof}

\begin{proof}[Proof of \Cref{cor:aggregate-spread}]
Take the essential supremum of \Cref{eq:fiberwise-bound} over $g\sim\nu_t$.
\end{proof}

\begin{proof}[Proof of \Cref{prop:approx-sufficiency-bound}]
Apply \Cref{eq:coarse-domain-pairwise}. Approximate sufficiency loses at most $\kappa$ and mismatch at most $\bar\varepsilon_t$; the two current cells lie in one coarse fiber, whose spread bounds the left side. Positive parts and essential suprema give the displays.
\end{proof}

\begin{proof}[Proof of \Cref{cor:primitive-kappa}]
Within a fixed coarse fiber and age pair, $\mathcal H$ oscillates by at most $L_t\Delta_t$. Taking the midpoint as $\bar h_t$ puts every value within half that oscillation, yielding the displayed admissible $\kappa$.
\end{proof}

\begin{proof}[Proof of \Cref{prop:constraint-amplification}]
Let $S$ be the continuation marginal-utility sum. The Euler equation $u'(c_a)=\beta p_a(1+r)S+\lambda_a(i,k)$ and the derivative $-u'(c_a)+\beta p_aRS$ give $\widehat R_a^P=(1+r)\WedgeTerm_a^{(i,k)}$, with $\WedgeTerm_a^{(i,k)}=[1-\lambda_a(i,k)/u'(c_a)]^{-1}\ge1$ and $\Qterm_a^{(i,k)}=1+r$. \Cref{prop:coarse-matching,prop:implementability} imply the claim.
\end{proof}

\begin{proof}[Proof of \Cref{prop:grid-convergence}]
After component normalization, the finite program is a closed system of edge-band inequalities. Assumptions (i) to (iv) give outer convergence of constraints, labels, and target weights, exclude persistent spurious constraints, and keep feasible sets compact. Assumption (v) gives nonemptiness and inner convergence, ruling out vanishing-cycle artifacts. Hence anchored feasible potential sets converge in Hausdorff distance. The target functional in \Cref{eq:finite-graph-objective} is continuous on the common compact envelope, so the interval follows by the maximum and minimum theorem.
\end{proof}


\clearpage
\begingroup
\singlespacing
\small



\endgroup

\end{document}